\documentclass{article}

\usepackage{microtype}
\usepackage{graphicx}
\usepackage{subfigure}
\usepackage{booktabs} 

\usepackage{hyperref}

\usepackage[accepted]{icml2024}

\usepackage{amsmath}
\usepackage{amssymb}
\usepackage{mathtools}
\usepackage{amsthm}

\usepackage[capitalize,noabbrev]{cleveref}

\theoremstyle{plain}

\usepackage{parskip}
\usepackage[utf8]{inputenc} 
\usepackage[T1]{fontenc}    
\usepackage{hyperref}       
\hypersetup{colorlinks,allcolors=blue,linktocpage=true}
\usepackage{url}            
\usepackage{booktabs}       
\usepackage{amsfonts}       
\usepackage{nicefrac}       
\usepackage{microtype}      
\usepackage{xcolor}         
\usepackage{amssymb}
\usepackage{amsmath}
\usepackage{mathtools}
\usepackage{amsthm, bbm, bm}
\usepackage{thm-restate, thmtools, enumitem}
\usepackage{soul, enumitem}
\usepackage{multicol}

\newtheorem{lemma}{Lemma}

\newtheorem{theorem}{Theorem}

\newtheorem{definition}{Definition}

\newtheorem{remark}{Remark}
\newtheorem{prop}{Proposition}

\DeclarePairedDelimiterX\brk[2]{\langle}{\rangle}{#1\,,\,#2} 
\DeclarePairedDelimiterX\Set[2]{\{}{\}}{#1 \;;\; #2} 

\newcommand{\argmax}{\operatornamewithlimits{arg\,max}}

\newcommand\independent{\protect\mathpalette{\protect\independenT}{\perp}}
\def\independenT#1#2{\mathrel{\rlap{$#1#2$}\mkern2mu{#1#2}}}
\newcommand\indep\independent

\newcommand{\EE}{\mathbb{E}}

\newcommand{\NN}{\mathbb{N}}

\newcommand{\PP}{\mathbb{P}}

\newcommand{\RR}{\mathbb{R}}

\newcommand{\Fcal}{\mathcal{F}}

\newcommand{\Hcal}{\mathcal{H}}
\newcommand{\Ical}{\mathcal{I}}

\newcommand{\Lcal}{\mathcal{L}}

\newcommand{\Pcal}{\mathcal{P}}

\newcommand{\Bal}{\begin{align}}
\newcommand{\Eal}{\end{align}}
\newcommand{\Beq}{\begin{equation}}
\newcommand{\Eeq}{\end{equation}}
\newcommand{\Bit}{\begin{itemize}}
\newcommand{\Eit}{\end{itemize}}
\newcommand{\Ben}{\begin{enumerate}}
\newcommand{\Een}{\end{enumerate}}

\newcommand{\Ba}{\begin{array}}
\newcommand{\Ea}{\end{array}}

\newcommand{\Bvec}{\left(\begin{array}{c}}
\newcommand{\Evec}{\end{array}\right)}

\newcommand{\Bmat}{\left(\begin{array}}
\newcommand{\Emat}{\end{array}\right)}

\newcommand{\Bol}{\begin{outline}}
\newcommand{\Eol}{\end{outline}}

\usepackage[textsize=tiny]{todonotes}

\icmltitlerunning{Peeking with PEAK: Sequential, Nonparametric Composite Hypothesis Tests for Means of Multiple Data Streams}

\begin{document}

\twocolumn[
\icmltitle{Peeking with PEAK: Sequential, Nonparametric Composite Hypothesis Tests for Means of Multiple Data Streams}




\begin{icmlauthorlist}
\icmlauthor{Brian Cho}{yyy}
\icmlauthor{Kyra Gan}{yyy}
\icmlauthor{Nathan Kallus}{yyy}
\end{icmlauthorlist}

\icmlaffiliation{yyy}{Department of ORIE, Cornell Tech, New York, NY, USA}

\icmlcorrespondingauthor{Brian Cho}{bmc233@cornell.edu}
\icmlkeywords{Machine Learning, ICML}

\vskip 0.3in
]



\printAffiliationsAndNotice{
} 

\begin{abstract}
We propose a novel nonparametric sequential test for composite hypotheses for means of multiple data streams. Our proposed method, \emph{peeking with expectation-based averaged capital} (PEAK), builds upon the testing-by-betting framework and provides a non-asymptotic $\alpha$-level test across any stopping time. Our contributions are two-fold: (1) we propose a novel betting scheme and provide theoretical guarantees on type-I error control, power, and asymptotic growth rate/$e$-power
in the setting of a single data stream; (2) we introduce PEAK, a generalization of this betting scheme to multiple streams, that (i) avoids using wasteful union bounds via averaging, (ii) is a test of power one under mild regularity conditions on the sampling scheme of the streams, and (iii) reduces computational overhead when applying the testing-as-betting approaches for pure-exploration bandit problems.  
We illustrate the practical benefits of PEAK using both synthetic and real-world HeartSteps datasets.
Our experiments show that PEAK provides up to an 85\% reduction in the number of samples before stopping compared to existing stopping rules for pure-exploration bandit problems, and matches the performance of state-of-the-art sequential tests while improving upon computational complexity. 
\end{abstract}

\section{Introduction}

Sequential experiments have important applications in various fields to optimize decision-making, including better content recommendation~\cite{bouneffouf2012contextual}, personalized learning~\cite{clement2013multi,cai2021bandit}, and optimized digital interventions~\cite{figueroa2021adaptive, forman2019can, liao2020personalized, piette2022artificial, trella2022designing, trella2023reward, yom2017encouraging}. 
Providing informative and anytime-valid feedback during sequential experiments could potentially lead to cost reduction through early stopping~\cite{liang2023experimental} and improved experiment outcomes.  
For example, in digital interventions,  prompt ``expert'' feedback has been shown to increase user engagement, resulting in better treatment outcomes~\cite{sharpe2017examining,yardley2016understanding}.
In this work, we aim to provide anytime-valid inference in sequential experiments for a task of interest (e.g., identifying the best arm/treatment in the bandit/reinforcement learning setting), where data collection adapts based on previously observed outcomes. 
The problem necessitates special inferential tools since naively repeatedly applying conventional inference such as $t$-tests at every time point will always at some point detect an effect even when none exists.


In settings where sample sizes are limited and multiple arms are to be compared, the ideal testing tools for sequential and possibly adaptive experiments should fulfill 
the following desiderata:
\begin{enumerate}[label={(\textcolor{blue}{\textbf{\alph*}})}, ref=\textbf{\alph*}]
    \item Providing a non-asymptotic $\alpha$-level test across any stopping time under nonparametric assumptions;\label{C1}
    \item Efficiently rejecting hypotheses that are incorrect across all potential distributions that satisfy our nonparametric assumptions;\label{C2}
    \item Enabling joint, composite hypothesis tests on multiple streams of data in a computationally tractable manner.\label{C3}
\end{enumerate}

The first tools for non-asymptotic, $\alpha$-level tests were proposed by \citet{Wald1945SequentialTO}, but these approaches rely on likelihood ratios, which may not be applicable in composite and nonparametric settings. 
To circumvent such issues, modern tests satisfying desideratum (\ref{C1}) use the notion of $e$-processes \cite{grünwald2023safe}.
It generalizes likelihood ratios to the nonparametric, composite setting \cite{ramdas2023gametheoretic} and
offers
an alternative to $p$-values that is more suitable for testing under optional stopping and continuation \cite{Vovk_2021, ramdas2022admissible}. 

Various sequential testing methods focus on asymptotic guarantees
\citep{ bibaut2022near,robbins1970boundary,waudby2021time,woong2023design}, where approximately valid (and powered) inference only starts after sufficiently large sample sizes. While these methods appeal to other desiderata such as asymptotic calibration of type-I errors and optimal power, they may not be appropriate when inference is desired for small sample sizes, reflected in our desideratum (\ref{C1}).

Recent works utilizing $e$-processes for non-asymptotic anytime-valid inference fall into two groups. The first involves Chernoff-based approach \cite{Howard_2021}, which relies on sub-$\psi$ assumptions (such as sub-gaussian/sub-bernoulli) combined with peeling \cite{Capp__2013} or mixture techniques \cite{de_la_Pe_a_2004}. The second involves the testing-by-betting framework \cite{waudbysmith2022estimating}.
In the bounded observation setting, both empirical~\cite{waudbysmith2022estimating} and theoretical~\cite{shekhar2023nearoptimality} evidence indicate that the latter outperforms the former. 
However, these approaches 
suffer from computational inefficiency -- computing confidence sequences and conducting composite tests require 
performing a grid search
over the hypothesis space of means, which grows in dimension as the number of data streams increases. Thus, while the current testing schemes proposed in \citet{waudbysmith2022estimating} satisfy both desiderata {(\ref{C1})}, {(\ref{C2})}, they fall short of satisfying  {(\ref{C3})}.

\textbf{Contributions.$\;$} In this paper, we provide a novel 
nonparametric sequential test for composite hypotheses on bounded means that aim to achieve desiderata \eqref{C1}-\eqref{C3}. We build upon the testing-by-betting framework and establish theoretical guarantees for our procedure. Our contributions are two-fold: (1) we propose a novel betting scheme, and provide theoretical guarantees on type-I error control, power, and asymptotic growth rate/$e$-power that have not been shown for other adaptive testing-by-betting methods in the setting of a single data stream; (2) we introduce PEAK, a generalization of this betting scheme to multiple streams, that (i) avoids using wasteful union bounds via averaging, (ii) is a test of power one under mild regularity conditions on the sampling scheme of the streams, and (iii) reduces computational overhead when applying the testing-by-betting approaches for pure-exploration bandit problems.

\textbf{Outline.$\;$} In Section \ref{sec:setup}, we define $e$-values and introduce the concept of $e$-power, which serves as the natural analogue of power in settings involving optional stopping and continuation. 
Next, we present the testing-by-betting framework by \citet{waudbysmith2022estimating}.
In Section \ref{sec:single_stream}, we introduce our testing procedure 
in the single arm/stream setting,  and provide theoretical results demonstrating that this test (i) controls type-I error, (ii) is 
of power 1, 
and (iii) forms convex/continuous confidence intervals.
In addition, we quantify the asymptotic suboptimality of our test.
In Section \ref{sec:multi_stream}, we provide a new test based on \emph{averaging evidence across arms}, and show that this testing approach maintains computational tractability while strictly dominating a union-bound approach. 
In Section~\ref{sec:experiments}, we
numerically test our sequential testing scheme for both single-arm and multi-arm settings on both simulated data and the mHealth HeartSteps dataset \cite{liao2020personalized}, demonstrating the efficacy of our approach. 



\section{Problem Formulation and Related Work}\label{sec:setup}
We consider the setting where we have $W$ arms/streams, indexed by $a \in \{1,...,W\}\equiv [W]$, with each arm corresponding to an unknown, independent distribution $P_a$ over $[0,1]$. We denote $\mu = [\mu_1, ..., \mu_W]$ as the means of the distributions $\{P_a\}_{a\in[W]}$. 
Let $X_t\in[0,1]$ denote the outcome that we observe at time $t$.
The data is collected in an online, sequential fashion: at time $t \in \NN$, we observe pair $(A_t, X_t) \in [W] \times [0,1]$. The choice of arm $A_t$ is sampled from a (potentially adaptive) sampling policy $\pi_t$, which may be unknown.
The sampling function $\pi_{t} \in \Delta^W$ is assumed to be $\Fcal_{t-1}$-measurable, where $\Delta^W$ is the probability simplex over $[W]$. For the entirety of this paper, we use $\Fcal_{0},\Fcal_{1},\dots$ as the canonical filtration, i.e., $\Fcal_{t-1} = \sigma((A_i, X_i)_{i=1}^{t-1})$, with $\Fcal_0$ as the empty sigma field. 

We assume the nonparametric setting, making no further assumptions about the data-generating process (DGP), allowing for continuous and discrete distributions over $[0,1]$. 
For each hypothesis $m_a \in [0,1]$ of arm $a$, let $\Pcal(m_a)$ denote the set of all distributions on $[0,1]$ with mean $m_a$. 
Each hypothesis vector $m = (m_1, ..., m_W) \in [0,1]^W$ denotes the composite hypothesis $\cap_{a \in [W]}\{P_a \in \Pcal(m_a)\}$, 
in which we fix the distribution means at $m$
but otherwise allow for arbitrary distributions over $[0,1]$ for each arm.\footnote{The minimal assumption necessary on the arm distributions for all results in this document can be found in Section 2, paragraph 2 of \citealp{waudbysmith2022estimating}.
It only requires stationary of the expectations for each arm conditional on the previous history.}
We use $R \subseteq [0,1]^W$ to denote the set of composite hypotheses of means (corresponding to regions of the hypothesis space) that we are interested in testing at  
time $t \in \NN$. 
We defer further discussion on $R$ to Section \ref{subsec:R_multi_stream}.

A sequential test $T_t(R)$ for a composite hypothesis $R$ is a binary-valued $\mathcal F_t$-measurable variable, with value 1 indicating rejection of $R$.
Our goal is to develop a computationally tractable sequential test $T_t(R)$ that maintains a $\alpha$-level type-I error rate, i.e., whenever $P_a\in\mathcal P(m_a)$ for all $a$ for some $m\in R$, we have
$$ 
\PP(\exists t\in \NN :T_t(R) = 1) \leq \alpha.$$


The test we developed will
follow
from the $e$-value and the testing-by-betting frameworks \cite{waudbysmith2022estimating}. We formally introduce these concepts below.

\subsection{$e$-Processes and Testing-by-Betting Framework} 
The notation of $e$-values, first introduced by \citet{Ville1939}, 
can be viewed 
as an alternative to $p$-values under optional stopping and continuation \cite{grünwald2023safe}. An $e$-value of hypotheses $\Hcal$ is defined as a non-negative random variable, $E$, that satisfies the following:
\begin{equation}
   \EE_P[E] \leq 1, \forall P \in \Hcal.\label{eq:e_var} 
\end{equation}

The resulting $\alpha$-level test from an $e$-value is simply thresholding $E$ against $1/\alpha$, with rejection if $E \geq 1/\alpha$. In our
setting, the $e$-value generalizes to the $e$-process:

\begin{definition}[$e$-process, Definition 1 of \citealt{grünwald2023safe}]\label{defn:e_process}
    We say that $E_t$ is a conditional $e$-variable (w.r.t. $(\Fcal_t)_{t\in \NN}$) for hypothesis $\Hcal$ if it is a $\Fcal_{t}$-measurable nonnegative random variable that satisfies $\EE_P[E_t|\Fcal_{t-1}] \leq 1$ a.s. for all $P \in \Hcal$. The product of conditional $e$-variables $K_t = \prod_{i=1}^t E_i$ is an $e$-process, meaning that for any potentially infinite stopping time $\tau$, $K_\tau$ is an $e$-value. 
\end{definition}

If $K_t$ is an $e$-process for hypothesis $\Hcal$, the $\alpha$-level sequential test is to reject $\Hcal$ if $K_t$ ever exceeds the threshold $1/\alpha$. The design of $e$-value-based tests directly leads to the interpretation of $K_t$ as \emph{evidence against the null}: we reject $\Hcal$ when $K_t$ is sufficiently large. The analogue of power for an $ e$-process-based test, called the growth rate or $e$-power \cite{vovk2024efficiency}, corresponds to this intuition.

\begin{definition}[$e$-power]\label{defn:e_power}
    For an $e$-value $E$ concerning the
    null hypothesis $P$, the growth rate of $E$ with respect to an alternative hypothesis $Q$ is given by $\EE_Q[\log(E)]$.
\end{definition}

In the case of 
$e$-process $K_t$, the $e$-power is the sum of the expected log values for each conditional $e$-variable $E_t$. Maximizing the $e$-power intuitively means 
maximizing the evidence against the null when the null is misspecified; under certain conditions, this directly corresponds to minimizing the expected stopping time in the i.i.d. single arm setting  \cite{grünwald2023safe, terschure2023anytimevalid}.
Moreover, for simple null hypothesis $P$ and simple alternative hypothesis $Q$ where
$Q$ is absolutely continuous with respect to $P$,
the likelihood ratio $dQ/dP$ maximizes the $e$-power (Lemma 2.1 of \citealp{vovk2024efficiency}), aligning directly with the same notion of power in classical parametric sequential tests  \cite{likelihood_ratio_opt}. 

\textbf{Testing-by-Betting Framework. $\;$} 
In the single-arm setting for bounded means, with null hypothesis $R=m$, the test we develop can be seen as a special case of the more general testing-by-betting framework \cite{waudbysmith2022estimating}, which constructs tests based on \emph{sizing bets}. 
Specifically, the testing-by-betting framework 
constructs the
$e$-process $M_t(m)$
as follows:
\begin{equation}\label{eq:ramdas_capital}
    M_t(m) = \prod_{i=1}^t \left(1+\lambda_t(m)(X_i - m) \right),
\end{equation} 
for some $\lambda_t(m) \in [-1/(1-m) , 1/m]$, where $\lambda_t(m)$ is $\Fcal_{t-1}$-measurable.\footnote{Let $1/m = \infty$ when $m=0$, $-\frac{1}{1-m} = -\infty$ when $m=1$.}  The sequence $\lambda_t(m)$ 
can be interpreted as bets against the belief of whether $\mu$ is above $m$, with
the sign indicating the direction of the belief and the magnitude indicating the strength of the belief
(i.e., how much of our current evidence $M_t(m)$ we are willing to put at stake).
We discuss our novel choice of $\lambda_t(m)$ in Section~\ref{sec:single_stream}, using the single-arm case as a guiding example.

The proposed testing-by-betting methods (i.e., choice of $\lambda_t(m)$) in \citet{waudbysmith2022estimating} suffer from computational
inefficiency
when generalized to the multi-armed case. In the single-stream setting, these tests require a search over the hypothesis space $[0,1]$, using grid search to rule out hypotheses $m$ sequentially. In the multi-stream setting, the size of the grid search grows exponentially with the dimension $W$. Our choice of sequential test, even with multiple streams, removes the need for grid search and allows various convex optimization approaches to be used.




\subsection{Hypothesis Testing of Means for Multiple Streams}\label{subsec:R_multi_stream}

In the multi-stream setting, 
we are interested in testing hypotheses
on the joint hypothesis space $[0,1]^W$. A common hypothesis 
of interest
is whether 
arm $a \in W$ is at least as large as all other arms, often referred to as the  best arm identification (BAI) problem \cite{Audibert2010BestAI}, i.e., 
\begin{equation}\label{eq:BAI}
R_{\text{BAI}}(a) = \{m \in [0,1]^W: m_a \geq m_i, \forall 
i \in [W]\setminus a
\}.    
\end{equation}
Other potential hypotheses of interest include determining if each stream's mean is above/below a certain threshold \cite{kano2018good}. For such objectives, the hypotheses 
form convex regions of the hypothesis space $[0,1]^W$. In Section~\ref{subsec:computation}, we
show 
our test enables computationally tractable testing for such convex regions.

Another work that directly considers joint sequential tests across arms is \citet{kaufmann2021mixture}, which uses a mixture-based approach to construct $e$-processes for a single arm and combines this evidence using the product of evidence across arms. 
While their objective aligns with ours,
the setting of their work 
differ drastically:
the parameters of the data streams are assumed to belong to 1-dimensional exponential families, 
rather than our nonparametric setting. 


\section{Tests for a Single Stream of Data}\label{sec:single_stream}
We first introduce our test in the setting where we have a single stream of data, i.e., $W=1$, with $m \in \RR$, and null hypothesis $R = m$. This serves as the base case for our approach. We define the \emph{capital process} $K_t(m)$:

\begin{definition}[Single-Arm Capital Process]\label{defn:one_arm_capital_process}
The capital process $K_t(m)$ and the corresponding log capital process $l_t(m)$ are defined as follows:
\begin{align}\label{eq:cap_process_single}
    K_t(m) &= \prod_{i=1}^t\left(1+\frac{(\hat{\mu}_{i-1} - m)(X_i - m)}{c}\right),\\
    l_t(m_a) &= \sum_{i=1}^t \log\left(1 + \frac{(\hat{\mu}_{i-1} - m)(X_i - m)}{c}\right),\label{eq:log_cap_process_main_body}
\end{align}
where $\hat\mu_i = m_a$ if $i=1$ and $\hat\mu_{i-1} = \frac{1}{i-1}\sum_{j=1}^{i-1} X_j$ for all $i > 1$. Moreover,  $c\geq 1/4$ is a constant. 
\end{definition}
Let $e_i(m) = 1 + \frac{(\hat\mu_{i-1}-m)(X_i-m)}{c}$. Then, \eqref{eq:e_var} implies that $e_i(m)$ is an $e$-variable. The capital process $K_t(m)$ can be written as
%
$K_t(m) = \prod_{i=1}^t e_i(m)$. By Definition \ref{defn:e_process},
the capital process 
$K_t(m)$
 forms
an $e$-process for hypothesis $m$. 
We note that 
$K_t(m)$, \eqref{defn:pkp}, is a special 
case of
the nonnegative martingale 
construction $M_t(m)$, \eqref{eq:ramdas_capital}, with $\Fcal_{t-1}$-measurable betting sequence $\{\lambda_t(m) = \frac{(\hat\mu_{t-1}-m)}{c}\}_{t\in\NN}$. 
This betting scheme
is 
intuitive:
we place our bet according to the sign of $\hat\mu_{t-1} - m$ (a plug-in estimate of the sign of $X_t - m_a$) and base the bet's magnitude on the absolute distance between $\hat\mu_{t-1}$ (our running estimate of $X_t$) and hypothesis $m$. 

We note that our choice of $\lambda_t(m) = \frac{\hat\mu_{t-1} - m}{c}$ differs from all existing betting schemes. The closest betting scheme to our choice of $\{\lambda_t(m)\}_{t\in \NN}$ is the AGRAPA method (Appendix B, \citealt{waudbysmith2022estimating}), which sets $\lambda_i(m) = \max(\frac{-l}{1-m}, \min\{\frac{l}{m}, \frac{\hat\mu_{i-1} - m}{\sigma_{i-1}^2 + (\hat\mu_{i-1} - m)^2}\} )$ for some $l \leq 1$ and a $\Fcal_{i-1}$-measurable estimate of the variance, $\hat\sigma_{i-1}^2$. We discuss the differences between our $\lambda$ scheme and AGRAPA in Prop. \ref{prop:convex} that enable computational benefits.

Because $K_t(m)$ is an $e$-process, it defines the following sequential test and confidence sequence.

\begin{definition}
    Let $T_t(m, \alpha) = 1$ denote the rejection of null (composite) hypothesis $\Pcal(m)$. The sequential capital process test is as follows:
    \begin{equation}\label{eq:test_single}
        T_t(m, \alpha) = \mathbf{1}\left[\max_{1\leq i \leq t}K_i(m) \geq 1/\alpha\right].
    \end{equation}
    The dual confidence sequence is respectively defined as
    \begin{equation}\label{eq:conf_int}
        C_t(\alpha) = \{m \in [0,1]: T_t(m, \alpha) < 1/\alpha\}.
    \end{equation}
\end{definition}

Theorem \ref{thm:type_1_error_correctness} below guarantees the correctness of this test. The probability of rejecting $m$ when $m = \mu$ is less than or equal to $\alpha$ for all $t \in \NN$ (proof in Appendix \ref{app:proof_thm1_correct}):
\begin{theorem}\label{thm:type_1_error_correctness}
    For all $c \geq 1/4$, $T_t(m, \alpha)$ defines a sequential test with type-I error $\alpha$ for null hypothesis $m$, i.e., if $\mu = m$ (equivalently, $P_1\in\mathcal P(m)$), $\PP(\exists t \in \NN: T_t(m) = 1) \leq \alpha$. 
\end{theorem}

This result is obtained by a direct application of Ville's inequality.
While Theorem \ref{thm:type_1_error_correctness} guarantees the correctness of our test when the null hypothesis $m = \mu$, it does not 
provide
guarantees to incorrect hypotheses. 
To address this, 
Theorem \ref{thm:plug_in_process_power_1} (proof in Appendix \ref{app:proof_thm_2}) states 
that for any $m \neq \mu$, $T_t(m,\alpha)$ is a test of power one, i.e., the stopping time of our test for an incorrect mean is finite. 

\begin{theorem}\label{thm:plug_in_process_power_1}
    For null hypothesis $m \neq \mu$, any $\alpha \in (0,1)$ and any $c > 1/4$, (i) the capital process  $K_t(m) \rightarrow \infty$ a.s. for any $m \neq \mu$, and (ii) the sequential test $T_t(m, \alpha) = \mathbf{1}[i \leq t : K_i(m) \geq 1/\alpha]$ is a test of power one, i.e., 
    $$\inf_{P \in \Pcal(\mu)}P(\exists t < \infty: T_t(m, \alpha)  = 1) = 1.$$  
\end{theorem}

We next verify that $C_t(\alpha)$, the confidence sequence defined in Equation \eqref{eq:conf_int},  is an interval by establishing the convexity of $K_t(m)$ below. Then a direct implication of Theorem \ref{thm:plug_in_process_power_1} is that $C_t(\alpha)$ shrinks to volume 0 as $t \rightarrow \infty$ for all $\alpha \in (0,1)$. 

\begin{prop}\label{prop:convex}
    For all $t \in \NN$, and for any sequence $(X_i)_{i=1}^t \in [0,1]^t$, $K_t(m)$ is strictly convex for all $m \in \RR$. Thus, there exists a unique minima $m^*$, which is the only $m \in \RR$ that satisfies the following equation: 
    $$\sum_{i=1}^{t} \frac{2m - (\hat\mu_{i-1}-X_i)}{c + (\hat\mu_{i-1} - m)(X_i - m)} = 0. $$
    Furthermore, the unconstrained minima $m^*$ is in $[0,1]$. 
\end{prop}

The convexity of our confidence sequence $C_t(\alpha)$ follows directly from Proposition \ref{prop:convex} (proof in Appendix \ref{app:proof_prop_1}):
at time $t=0$,
$C_0(\alpha) = [0,1]$; 
for $t \in \NN$, each test results in a convex set; repeated intersections $C_t(\alpha)$ is also convex. Note that this property does not generally hold in testing-by-betting schemes by
\citet{waudbysmith2022estimating}.

\begin{remark}
For example, the AGRAPA betting scheme, despite its similar form to our choice of betting scheme, is not guaranteed to form convex confidence sets (\citealp{waudbysmith2022estimating}, Appendix E.4). This requires one to test each $m \in [0,1]$ to construct confidence sets (or each $m \in R$ to determine whether one can reject a region of interest $R$). Our choice of $\lambda$ remedies this issue by (1) forming strictly convex confidence sets with respect to $m$, avoiding the need to test a grid of values over $[0,1]$ when constructing confidence sets, and thus (2) enabling a simple characterization of the minima $m$ in a convex region $R$, avoiding the need to test a fine grid of $m \in R$. These two benefits reduce the computational burden of our proposed testing method PEAK, a generalization of this simple betting scheme $K_t(m)$, discussed in Section \ref{sec:multi_stream}.
\end{remark}

%
%
While Theorems \ref{thm:type_1_error_correctness} and \ref{thm:plug_in_process_power_1} guarantee that $\mu$ will be identified in our test, they do not provide additional insights. 
To better understand the behavior of our test, 
in Section~\ref{subsec:limiting_e_power}, we characterize the asymptotic $e$-power/growth rate (Definition~\ref{defn:e_power}) of our process, $K_t(m)$, 
which defines the expected linear rate of increase of $l_t(m)$ as $t \rightarrow \infty$. 
We discuss the selection of parameter $c$ in Lemma~\ref{lem:c_small_optimal}.
In Lemma~\ref{lem:expection_orcl}, we establish that the worst-case instance of $e$-power
is strictly positive for $m\neq \mu$.
Finally, 
given that our test is nonparametric, we compare the asymptotic growth rate of our test in the parametric setting to the best-achievable growth rate in Lemma~\ref{lem:orcl_growth_rate}, establishing that the $e$-power of our test is near-optimal in the parametric setting.



\subsection{Asymptotic $e$-Power/Growth Rate} \label{subsec:limiting_e_power}
The asymptotic $e$-power of $K_t(m)$ is defined as
follows:
\begin{definition}
    Let $P$ be a distribution with mean $\mu$, i.e., $P \in \Pcal(\mu)$. We define the asymptotic growth rate of $K_t(m)$
    as $G(c,m, P) = \lim_{t\rightarrow \infty} \EE_P[\log(1+\frac{(\hat\mu_{t-1} - m)(X - m)}{c})] =  \EE_P[\log(1+\frac{(\mu - m)(X_t - m)}{c})]$.\footnote{The interchange of limits and integrals in the last equality is valid by Leibnitz's rule.}
\end{definition}
%

In Lemma~\ref{lem:c_small_optimal} (proof in Appendix \ref{app:proof_lem_2}), we establish the relationship between variable $c$ and asymptotic growth rate $G(c,m,P)$. This provides a natural choice of $c \in [1/4, \infty)$.
\begin{lemma}[Asymptotic Growth Rate as a Function of $c$]\label{lem:c_small_optimal}
    For any $m \in [0,1]$ and any alternative hypothesis $P$, the growth rate $G(c, m, P)$ monotonically decreases with $c$. 
\end{lemma}
While setting $c$ to its minimal value maximizes the asymptotic growth, we recommend setting $c=1/4 + \gamma$, where $\gamma$ is a small constant such as $\gamma = 0.01$ for the sake of stability.\footnote{
Consider the scenario with $c=1/4$, where the estimated mean from previous observations $\hat\mu_{i-1} = 1$, $m = 0.5$, and the next observed data point $X_i = 0$. Then, $K_t(m) = 0$. For any time after $t$, $K_t(m) = 0$, meaning that the hypotheses $m=0.5$ cannot be rejected for \emph{any} time in the future.} 



Next, in Lemma~\ref{lem:expection_orcl} (proof in Appendix \ref{app:proof_lemma_3}), we establish that the worst-case growth rate of $G(c, m, P)$ is strictly positive for hypotheses $m\neq \mu$, and increases as a function of the absolute difference between the null hypothesis $m$ and the true mean $\mu$. 



\begin{lemma}[Worst-Case  Positive Asymptotic Growth Rate]
\label{lem:expection_orcl}
 For any fixed $m \in [0,1], \mu \in [0,1]$, the worst-case instance $P^*:=\inf_{P \in \Pcal(\mu)} G(c,m,P)$ is $P^* \equiv \text{Bern}(\mu)$.

Furthermore, for any $m \in [0,1]$, and any distribution $P$ over $[0,1]$ with mean $\mu$, $c \geq 1/4$, the worst-case growth rate is strictly nonnegative, i.e.,
    $\inf_{P \in \Pcal(\mu)}G(c,m,P) \geq 0.$
    In particular, (1) $G(c, m, P)$ is equal to 0 only if $\mu = m$, and (2) $G(c, m, P)$ monotonically increases w.r.t. $|\mu - m|$.
\end{lemma}





Finally, our test is 
nonparametric. When compared with the likelihood ratios of \citet{Wald1945SequentialTO}, naturally, one would expect our worst-case asymptotic growth rate to deviate significantly from the best achievable parametric growth rate. However, Lemma \ref{lem:orcl_growth_rate} (proof in Appendix \ref{app_proof_lem_4}) shows that our test achieves near-optimal asymptotic growth rate 
in the simple-vs-simple Bernoulli setting, an instance of our worst possible growth rate. Let $\mu \in (0,1)$. We have:
\begin{lemma}[Relative Growth Rate in Worst-Case Setting]\label{lem:orcl_growth_rate}
    When testing a Bernoulli distribution $X_i \sim \text{Bern}(\mu)$, the ratio of $G(c,m,P)$ to the best achievable growth rate in this simple-vs-simple sequential test is given by:
\begin{align}\notag
        f(c, m, \mu) = \frac{\log((1+\frac{(\mu - m)(1-m)}{c})^\mu (1-\frac{(\mu - m)m}{c})^{1-\mu})}{\log((\frac{\mu}{m})^\mu(\frac{1-\mu}{1-m})^{1-\mu} ) }.
    \end{align}
\end{lemma}

To get a better sense of the ratio $f(c, m,\mu)$, we plot both $G(c,m,P)$ and $f(c,m,\mu)$ in Figure \ref{fig:growth_rate_comparisons}.
\begin{figure}[t!]
    \centering
    \includegraphics[width=0.9\linewidth]{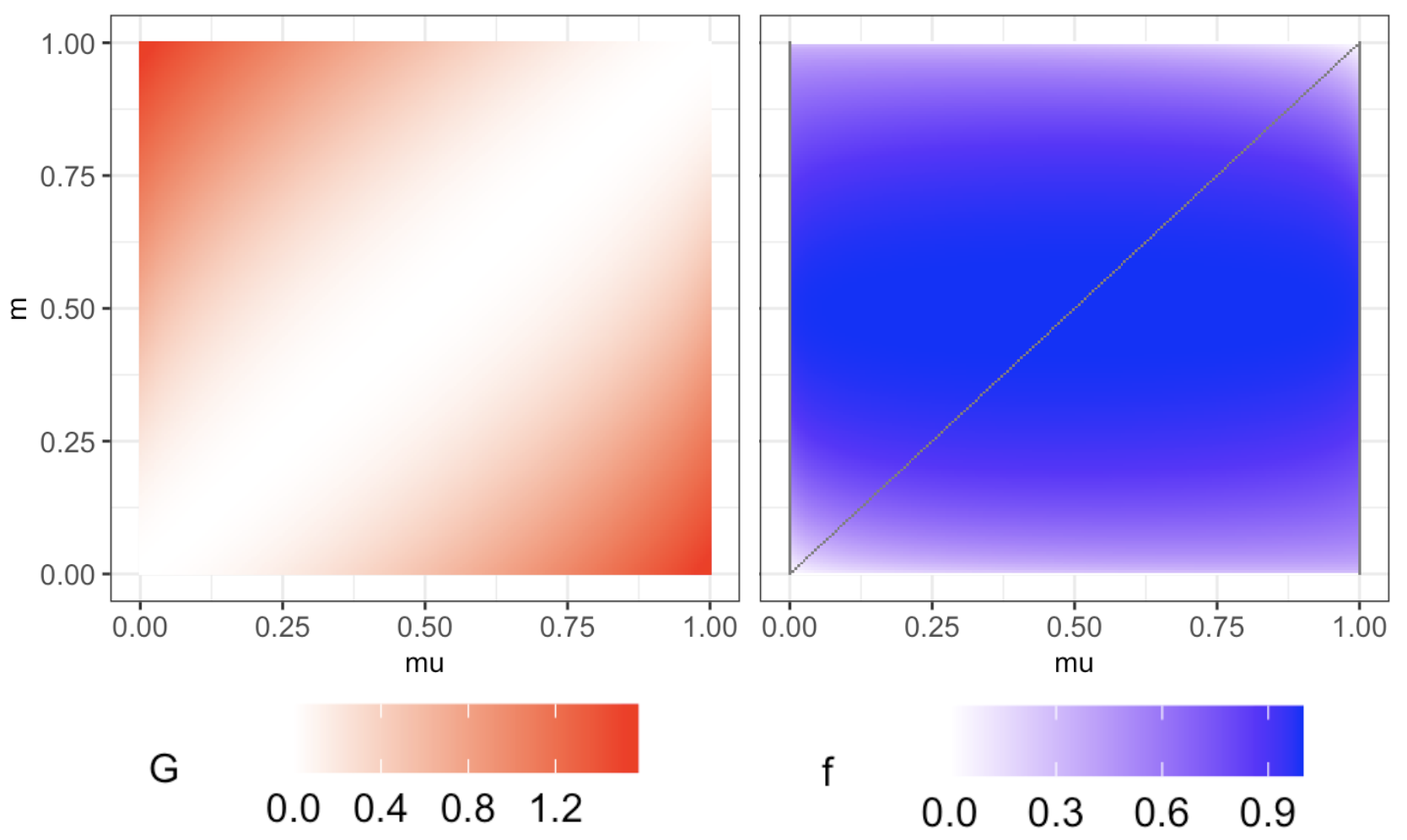}
    \vspace{-10pt}
    \caption{Growth rate visualization for different ground truth $\mu$ and hypothesis $m$ combinations under Bernoulli distributions, $P \equiv \text{Bern}(\mu)$, and $c=0.26$.  Left:
    Asymptotic growth rate of $K_t(m)$,
    $G(c, m , P)$, for null hypothesis $m$, with darker colors representing larger growth rates.
    Right: The \emph{ratio} of $G(c, m, P)$ to the best achievable growth rate,
    $f(c, m , P)$, for null hypothesis $m$, with darker colors representing larger ratios.
    }
    \label{fig:growth_rate_comparisons}
\end{figure}
We observe that as the distance between $\mu$ and $m$
decreases, $f(m, \mu)$ increases rapidly in Figure~\ref{fig:growth_rate_comparisons} (right), meaning that $K_t(m)$ preserves the best-possible expected growth rate in difficult instances (i.e., the null hypothesis $m$ is close to the alternative hypothesis $\mu$). For instances where our growth rate demonstrates suboptimality (i.e., $m \approx 0$ or $m \approx 1$), we observe that $G(c,m,P)$ is still large (white/orange regions of Fig. \ref{fig:growth_rate_comparisons} left). Lemma \ref{lem:orcl_growth_rate} and Figure \ref{fig:growth_rate_comparisons} demonstrate that for incorrect hypotheses close to the true mean, the asymptotic growth rate of our process $K_t(m)$, in the worst case, is near-optimal with respect to the best-achievable growth rate. 

\section{Testing in the Multi-Stream Setting}\label{sec:multi_stream}
In this section, we generalize the results of the single-stream/arm setting to $W$ streams/arms, where $W \in \NN$, and demonstrate how to test for a general null composite hypothesis $R$ that is convex in the hypothesis space $m \in [0,1]^W$. Let $N_t(a) = \sum_{i=1}^{t} \mathbf{1}[A_i = a]$ and $\hat\mu_{t}(a) =  \sum_{i=1}^t \mathbf{1}[A_i = a]X_i/N_t(a)$. We first define our product capital process and the test for the case $R = m$ below:
\begin{definition}\label{defn:pkp}
    The capital process $K_t^a(m_a)$ for each stream/arm and the corresponding joint capital process $E_t(m)$ are defined as follows:
\begin{align}
K_t^a(m_a) &= \prod_{i=1}^t
\bigg( 1 - \mathbf{1}[A_t = a] +  \mathbf{1}[A_t = a]\cdot
 \notag\\
&\quad \left(1+\frac{(\hat\mu_{i-1}(a) - m_a)(X_i - m_a)}{c}\right)\bigg)\label{eq:mult_stream_single_arm_capital}\\
E_t(m) &= \frac{1}{W}\sum_{a \in [W]} K_t^a(m_a).\label{eq:prod_capital_process}
\end{align}
 The resulting $\alpha$-level test for hypothesis $m$ is as follows:
 \begin{equation}\label{eq:joint_cap_process}
     T_t(m, \alpha) = \mathbf{1}\left[\max_{1 \leq i\leq t} E_t(m) \geq 1/\alpha\right].
 \end{equation}
 
\end{definition}

This process generalizes our single-arm setting by averaging evidence across each arm for any hypothesis $m$, which gives rise to our procedure's name as \emph{(P)eeking with (E)xpectation-based, (A)veraged (K)apital}.


\paragraph{Why should we average the evidence?} The naive way to generalize our sequential tests is to use union bounds: we divide $\alpha$ by $W$ and reject the hypothesis at time $t$ if each component $K_t^a(m_a)$ exceeds the $W/ \alpha$ threshold. In Lemma \ref{lem:union_bad} (proof in Appendix~\ref{app:proof_lem_5}), we establish that the product test strictly dominates the union bound test.

\begin{lemma}\label{lem:union_bad}
    Let $T'_t(m,\alpha)$ be the test under union bounding:
    $$T'_t(m, \alpha) = \prod_{a\in [W]}\mathbf{1}\left[ \max_{1 \leq i \leq t} K_i^a(m_a) \geq  W/\alpha \right].$$
    For all $t \in \NN$, any $m \in [0,1]^W$ rejected by union bounding test $
    T'_t(m, \alpha)
    $ is also rejected by the test $T_t(m, \alpha) = \mathbf{1}\left[ \max_{1 \leq i \leq t}E_i(m) \geq 1/\alpha\right]$ for any realization of the data. 
\end{lemma}

An immediate consequence of Lemma~\ref{lem:union_bad} is that for any realization of the data, the confidence sequences formed by inverting the test $T_t(m,\alpha)$ is a \emph{strict subset} of the confidence sequences formed by union bounding across arms. 

\subsection{Theoretical Guarantees for Joint Capital Process}
Let $\pi_t: H^{t-1}\rightarrow \Delta^{W}$ be a sampling scheme, where $H^{t-1}=\{A_1, X_1, ..., A_{t-1}, X_{t-1}\}$.
This process $E_t(m)$ shares the same theoretical guarantees of the single arm setting for 1) type-I error control for all
$\pi_t$,
and 2)  power under mild assumptions of
$\pi_t$. 
\begin{theorem}[Correctness of Joint Capital Test]\label{thm:product_theorem_correctness}
    For $m = \mu$, $c \geq 1/4$, for any sampling scheme $\pi_t$, 
    the test $T_t(m, \alpha)$ 
    has type-I error probability 
    $1-\alpha$.  
\end{theorem}

Theorem \ref{thm:product_theorem_correctness}  (proof in Appendix \ref{app:proof_thm_3}) provides type-I error control for any possible sampling scheme $\pi_t$, and maintains the correctness of our test. To show that
$T_t(m,\alpha)$ is
a test of power one for any $m \neq \mu$, we require
a relatively mild condition as stated in
Theorem \ref{thm:product_capital_diverges} (proof in Appendix \ref{app:proof_thm_4}). 

\begin{theorem}[Test of Power 1]\label{thm:product_capital_diverges}
    For any $m \neq \mu$, $c > 1/4$, let $\Ical(m) = \{a \in [W]: m_a \neq \mu_a\}$. Assume that the chosen sampling scheme $\pi_t$ satisfies the following: there exists $a \in \Ical(m)$ s.t. $N_t(a) = \sum_{i=1}^t \mathbf{1}[A_t = a] \rightarrow \infty$ as $t\rightarrow \infty$. Then, $T_t(m, \alpha)$ is a test of power 1, i.e., $\inf_{P \in \Pcal(\mu)}P(\exists t < \infty: T_t(m, \alpha)  = 1) = 1$.
\end{theorem}

Intuitively, the condition of Theorem \ref{thm:product_capital_diverges} states that for an incorrect hypothesis $m\neq \mu$, we must pull at least one misspecified arm infinitely often (not necessarily the same arm across data realizations) as $t \rightarrow \infty$. Simple examples of sampling schemes that satisfy this property include $\pi_t(a) = 1/W$ or bandit sampling schemes such as $\epsilon$-greedy or action elimination \cite{jamieson}.

\subsection{Testing Composite Convex Hypotheses}\label{subsec:computation}
The key benefit of our testing approach, as compared to using the betting schemes of \citet{waudbysmith2022estimating}, is computational tractability for composite null tests that define convex regions in the hypothesis space, $[0,1]^W$. To reject a region $R$ of the hypothesis space, 
we require that every hypothesis $m \in R$ in the region is rejected. In terms of our test, the criterion to reject a region $R$ is as follows:
\begin{equation}\label{eq:joint_test}
    T_t(R, \alpha) = \mathbf{1}\left[\min_{1 \leq i \leq t, m \in R} E_i(m) \geq 1/\alpha \right].
\end{equation}
For such composite tests, it is crucial to find
the minima of $E_t(m)$ within region $R$ efficiently.
We note that a straightforward extension of certain testing-by-betting schemes\footnote{One such scheme is AGRAPA (Appendix B.4. in \citealp{waudbysmith2022estimating}), which can form non-convex confidence regions.} 
by \citet{waudbysmith2022estimating} 
to the multi-arm setting would involve a grid search over the region $R \subset [0,1]^W$ to test hypothesis $R$. Because our composite test across joint means is simply the average of $K_t^a(m_a)$, it inherits the properties described in Proposition \ref{prop:convex}, such as strict convexity and a unique global minima $\tilde{m} \in [0,1]^W$. 
Consequently,
the global minima of $E_t(m)$ is characterized by the following:

\begin{prop}[Global Minima of $E_t(m)$.]\label{prop:global_min}
    For time $t$, the vector $m^* = [m_1^*, ... ,m_W^*] \in [0,1]^W$ that minimizes $E_t(m)$ is 1) unique, and 2) characterized component-wise by the following condition:
    $$\sum_{i=1}^t\sum_{a \in [W]} \mathbf{1}[A_t = a] \frac{2m_a^* - \hat\mu_{i-1}(a)-X_i}{c + (\hat\mu_{i-1}(a) - m_a^*)(X_i - m_a^*)} = 0.$$
\end{prop}

To find the minima within region $R$, we project our global minima $m^*$ as characterized in Proposition \ref{prop:global_min}.
Below,
We provide two 
examples of this projection: 
threshold identification (THR) and best-arm identification (BAI), and provide visualizations in Figure \ref{fig:thr_bai_viz} in Appendix \ref{app:additional_experiments}.

\paragraph{Example 1: THR.}
For a threshold $\xi \in [0,1]$, each arm $a$ in the threshold problem $\text{THR}$ has two corresponding hypothesis regions: $R_{\text{THR}}^0(a) = \{m \in [0,1]^W: m_a < \xi \}$, and $R^1_{\text{THR}}(a) = [0,1]^W \setminus R^{\text{THR}}_0(a)$. 

To test hypothesis $R_{\text{THR}}^0(a')$ at time $t$, we 
check if $E_t(m^0)$ exceeds the $1/\alpha$ threshold, where entries of $m^0$ are given by:
$$m^0_a = \begin{cases} m_a^* & \text{if } a \neq a'  \\ \min(m_a^*, \xi) & \text{if } a = a' \end{cases}. $$
To test hypothesis $R_{\text{THR}}^1(a')$ at time $t$, we 
check if $E_t(m^1)$ exceeds the $1/\alpha$ threshold, where entries of $m^1$ are defined as:
$$m^1_a = \begin{cases} m_a^* & \text{if } a \neq a'  \\ \max(m_a^*, \xi) & \text{if } a = a' \end{cases}. $$
The vectors $m^0$ and $m^1$ are the projections of the global minimizer $m_a^*$ for regions $R_{\text{THR}}^0(a')$, $R_{\text{THR}}^0(a')$, respectively, and only involve thresholding the mean of arm $a'$ based on $\xi$. Below, we consider a slightly more complicated example, BAI, which has constraints defined between pairs of arms.  


\paragraph{Example 2: BAI.} For best-arm identification, each arm $a$ has a single hypothesis, as defined in Equation \eqref{eq:BAI}. Proposition \ref{prop:BAI_min} provides the conditions that define the minima of $E_t(m)$ for region $R_{\text{BAI}}(a)$ at time $t$.

\begin{prop}[Minima Condition for BAI Partition]\label{prop:BAI_min}
Let $R_{\text{BAI}}(a')$ be the set defined in Equation \eqref{eq:BAI}. Let $\Tilde{m}\in [0,1]^W$ be the global minimizer of $E_t(m)$. Then, $\min_{m \in R_{\text{BAI}}(a')}E_t(m)$ is obtained at solution $m^*$, given by:
\begin{equation*}
    m^*_a = \begin{cases} \tilde{m}_a \quad &\text{if } \Tilde{m}_a < q\\ q \quad & \text{if } \Tilde{m}_a \geq q \end{cases},
\end{equation*}
where $q \in [0,1]$ is a constant that satisfies 
\begin{align*}
    &\sum_{a \in W: \Tilde{m}_a \geq q} \Bigg(\gamma(m^*,a) \\
    \quad &\sum_{i=1}^t  \mathbf{1}[A_t = a] \frac{2q - X_i - \hat\mu_{i-1}(a)}{c + (\hat\mu_{i-1}(a)-q)(X_i - q)}\Bigg) = 0,
\end{align*}
where $\gamma(m,a) = \prod_{i=1}^t ( \mathbf{1}[A_i = a](1 + (X_i - m_a)(\hat\mu_{i-1}(a) - m_a)/c)$ $ + (1-\mathbf{1}[A_i = a]) )$.
\end{prop}

Using this characterization, analysts can find the $m \in R_{\text{BAI}}(a')$ that achieves the minimum value within $R_{\text{BAI}}$ by solving a single equation, as opposed searching across the entire region with grid values.

Beyond characterizing the constrained minima for specific hypothesis sets, finding the minimum value of $E_t(m)$ over any generic convex region $R \subset [0,1]^W$ defines a standard convex optimization problem. Thus, this allows analysts to (1) use the plethora of available convex optimization tools \cite{boyd2004convex} beyond grid-search style approaches, and (2) enables a direct characterization of the region-specific minima through KKT conditions, as done above for the BAI problem in Proposition \ref{prop:BAI_min}.

\begin{figure*}[ht!]
\label{fig:conf_seq_widths}
    \centering    \includegraphics[scale=0.41]{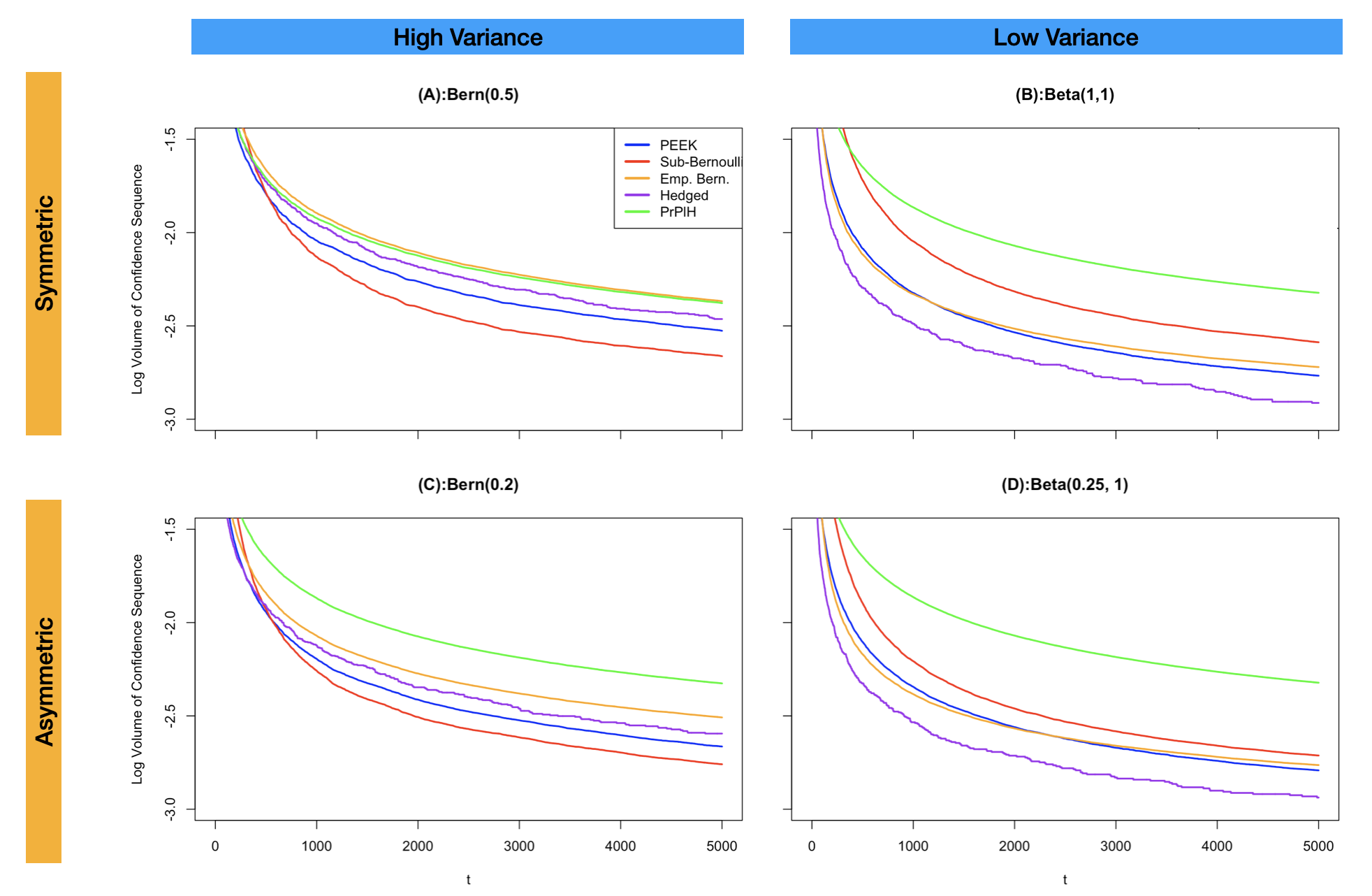}
    \vspace{-10pt}
    \caption{Log-widths of the confidence sequences across 5000 time steps. Curves represent the average width across 30 simulations.}
    \label{fig:conf_int_width}
\end{figure*}

\section{Empirical Results}\label{sec:experiments}

In this section, we provide empirical comparisons for our test that demonstrate (i) in the single-arm case, our test provides robust confidence sequences that verify the intuitions of Lemma \ref{lem:orcl_growth_rate}, (ii) in the multi-stream setting with adaptive sampling, our test provides improved empirical performance to verify the composite hypothesis tests corresponding to threshold identification and best arm identification, and (iii) a stylized case-study using the mobile Health dataset HeartSteps \cite{liao2020personalized}. 
We use $\alpha = 0.05$ for our simulations and use $\alpha = 0.3$ in our case study.\footnote{We note that our choice of $\alpha = 0.3$ is due to the lack of strong signal in this dataset \cite{liao2020personalized}.} For all experiments, we set $c = 0.26$\footnote{Code available at \url{https://github.com/brianc0413/PEAK}.}.

\subsection{Synthetic Experiments}\label{subsec:synthetic}

\paragraph{Experiments in the Single-Stream Setting}
We compare the confidence sequences generated by our test to existing anytime-valid confidence intervals. We use the confidence sequence as a proxy for both the \emph{power} and \emph{growth rate} of our test. We compare our approach with four other nonparametric sequential tests 
that directly apply to bounded means: 
1) Chernoff-based  
\emph{sub-Bernoulli} (Sub-B) confidence sequences \cite{Howard_2021},\footnote{as implemented using the package \texttt{confseq} \cite{confseq}.}
2) \emph{empirical Bernstein} (Emp. Bern.) style confidence sequences (\citealt{waudbysmith2022estimating}, Thm. 2),
3) Hoeffding-based confidence sequences (PrPlH, \citealt{waudbysmith2022estimating}, Prop. 1), and 4) the hedged capital process confidence sequence (Hedged, \citealt{waudbysmith2022estimating}, Theorem 3, Eq. 26). For Hedged, we use a grid of 100 points for all simulations. We provide experiment details and additional experiments with nonstandard distributions in Appendix \ref{app:additional_experiments}. 

\textbf{Multi-Streams of Data, Adaptive Sampling\;\;}
We test multiple streams by 1) introducing adaptive sampling schemes that vary over time, and 2) requiring $\alpha$-level guarantees across multiple parameters of interest. We set $W=4$. In the Bernoulli setting (high variance case), the distribution of arm $a$ is  $P_a \equiv \text{Bern}(0.15 + 0.14a)$. In the Beta setting (low variance case), the distribution of arm $a$ is $P_a \equiv \text{Beta}(1,  \frac{0.85 - 0.14a}{0.15 + 0.14a} )$. In both cases, $\mu_a = 0.14a + 0.15$. For both THR and BAI, \citealp{kano2018good} and \citealp{jamieson} (Eq. 2)  propose a stopping criterion that uses sequential tests that union bound over each arm. We denote their proposed stopping criterion as \emph{Base}, and additionally test the two best-performing tests/confidence sequences (Sub-B, Hedged) in the single-arm case by using a similar union-bound approach. We refer to Appendix \ref{app:additional_experiments} for the exact adaptive sampling algorithm used for each bandit identification problem, and further experiment results, including simulations with nonstandard data distributions. 

\begin{table}[t]\label{table:THR_shortened}
\centering
\begin{tabular}{rrr}
  \hline
Type & Stopping Rule & Stopping Time \\ 
  \hline
  & PEAK & {725.81 $\pm$ 218.49}\\ 
Beta   & Base & 4534.31 $\pm$ 482.94 \\ 
    & Sub-B & 1618.32 $\pm$ 229.90\\ 
    & Hedged & \textbf{479.16 $\pm$  205.97}\\ 
    \hline 
 & PEAK & \textbf{1678.08$\pm$ 666.39} \\ 
Bern &   Base  &  4795.51 $\pm$ 1101.09  \\
  &   Sub-B  & 1686.81 $\pm$ 456.08 \\
  & Hedged & 2241.18 $\pm$ 1092.49 \\ 
\hline  
\end{tabular}
    \caption{Stopping Times for THR. Stopping times represent the first time that we have classified all 4 arms as above or below $\xi= 0.5$. We report averaged stopping times and one standard error over 100 simulated sample paths. Across all simulations, each stopping time results in the correct conclusion (i.e., rejects the incorrect hypothesis).}
\end{table}

\begin{table}[t]\label{table:BAI}
\centering
\begin{tabular}{rrr}
  \hline
Type & Stopping Rule & Stopping Time \\ 
  \hline
  & PEAK & 708.52 $\pm$ 266.34 \\ 
Beta   & Base & 4686.1 $\pm$ 565.24 \\ 
    & Sub-B & 3175.3 $\pm$ 905.70\\ 
    & Hedged & \textbf{500.12 $\pm$ 214.08}\\ 
    \hline 
 & PEAK & \textbf{1318.14 $\pm$ 489.29} \\ 
Bern &   Base  &  4631.66 $\pm$ 896.51 \\
  &   Sub-B  & 4322.36 $\pm$ 3638.25\\
  & Hedged & 1734.64 $\pm$ 858.66 \\ 
\hline  
\end{tabular}
    \caption{Stopping Times for BAI. We report averaged stopping times and one standard error over 100 simulated sample paths. For all simulations, regardless of method, the correct best arm was identified.  }
\end{table}

\textbf{Discussion of Simulation Results\;\;}
In the single-arm case, Figure \ref{fig:growth_rate_comparisons} plots the log-width of the confidence sequences across time. The log-width of our confidence interval suggests that PEAK provides a \emph{robust testing procedure}, even in the single-arm case. Across any combination of variance settings and symmetry, our method performs the second-best across all tested approaches. In the high variance Bernoulli setting, Sub-B performs the best, as the sub-Bernoulli condition is tight with respect to the tested distribution. Compared to all other approaches, our test is closest in performance, corroborating the worst-case optimality results of Lemma \ref{lem:orcl_growth_rate}. For the low-variance case, Hedged outperforms all other methods, but is closely followed by Emp. Bern. and PEAK. 
We additionally note that PEAK outperforms all benchmarks under mixture distributions as illustrated in Figure~\ref{fig:mixture_conf_int}.

In the multi-armed case, our test performs comparably to the state-of-the-art testing-by-betting method, Hedged, and outperforms the proposed stopping criterion and Sub-B. For both THR and BAI, the stopping times for PEAK are either the lowest or second lowest on average empirically. For THR (Table \ref{table:THR}), the performance of our approach closely mirrors its performance in the single-arm case, where PEAK performs the second-best for both high and low variance settings, and Hedged and Sub-B perform the best for Beta (low variance) and Bernoulli (high variance) arms, respectively. For BAI,
the testing-by-betting methods (Hedged and PEAK) drastically outperform Base and Sub-B as a stopping criterion.
For the high variance, Bernoulli setting, PEAK outperforms Hedged, while Hedged outperforms PEAK in the low variance Beta setting.

\textbf{Runtime\;\;} 
To stress the key benefits of our approach, we test the best-arm hypothesis in the Bernoulli setting across a fixed $T=2000$ time horizon (one test every two time steps), with no stopping, and report the average runtime for the fixed-horizon test over 100 simulations in Table \ref{table:runtimes}. The runtimes for PEAK and Hedged are only comparable 
when the grid size is 100
(i.e., $10^{-2}$ fidelity) for Hedged. For larger grid sizes, such as 500 evenly spaced points, the runtime of Hedged is at least double that of PEAK, despite having far worse fidelity. The union-bound approach of using confidence intervals with the Hedged betting scheme scales linearly with the fidelity of the grid and becomes exponential if one were to 
average across arms, as we do with PEAK.  This demonstrates that even with a relatively small number of streams, our approach drastically improves upon the computational requirements of existing betting schemes, such as Hedged. 

\begin{table}\label{table:runtimes}
\centering
\begin{tabular}{rrrr}
  \hline
Method & Fidelity & Runtime (s)  \\ 
  \hline
  PEAK & $1\times 10^{-10}$ & 5.21 $\pm$ 0.23 \\ 
    \hline 
 &  $1 \times 10^{-2}$  & 3.72 $\pm$ 0.19\\ 
Hedged  &  $5 \times 10^{-3}$ & 6.982 $\pm$ 0.17  \\
  &    $2.5 \times 10^{-3}$  & 14.50 $\pm$ 0.65 \\
\hline  
\end{tabular}
\caption{Runtimes Times for BAI Hypothesis Test, over a Fixed Horizon $T=2000$, with tests at every other time point. We report averaged runtimes and one standard error over 100 simulated sample paths.  }
\end{table}

\subsection{Case Study with Mobile Health Data} We use the HeartSteps dataset \cite{liao2020personalized} as a case study for our approach. The HeartSteps dataset was collected using a data-adaptive sampling scheme during a mobile health intervention with the HeartSteps app and contains user-level information on 37 users: (1) $A_{i,t}$ a binary intervention received at time $t$ for user $i$, and (2) $X_{i,t}$, the number of user $i$'s steps within a short interval after the intervention. In our case study, we investigate the value of $\tau_i$, the first time in the trial in which we can determine whether $\mu_i(1) \geq \mu_i(0)$ (or vice versa) with 70\% confidence. 
We make the following assumptions for the DGP: (1) each user $i$ has an associated mean vector $\mu_i$, where for all $t$, $\EE[X_{i,t} | \Fcal_{t-1}, A_{i,t}=a] = \mu_{i}(a)$, and (2) the number of steps that can be taken are bounded above for all users. These two assumptions are sufficient for the type-I error guarantees of PEAK. We apply the BAI termination condition for $W=2$ for each $i$ in the dataset using PEAK, Sub-B, and Hedged, the three most empirically promising testing schemes from Section \ref{subsec:synthetic}, to obtain $\tau_i$ (if it occurs before the horizon of the data) for each user $i$.

In our experiments, we observe that
Sub-B fails to conclude that the intervention is either helpful or harmful for any user $i$ throughout the entire horizon of the trial. In contrast, both testing-by-betting methods are able to conclude that a treatment is helpful/harmful before the end of the trial, to varying degrees. Hedged concludes that the digital intervention provides no benefit and may be harmful (i.e., $\mu_{i}(1) \leq \mu_{i}(0)$) for User 22 using just half of the samples collected over the trajectory ($\tau_{22} = 4932$ over a horizon of 9072). 
For all other patients, Hedged fails to conclude that the intervention is either helpful or harmful. PEAK outperforms both approaches: it provides stopping times $\tau_i$ for 8 different patients before the end of the recorded data, identifying 8 different patients who all benefit from the treatment (i.e., $\mu_i(1) \geq \mu_i(0)$). If the analyst were to stop trials for the 8 users at their respective stopping times, they would have saved 18,563 total samples, reducing the total number of samples by roughly 10\%. This simple case study demonstrates PEAK's potential to reduce experimental costs and provide resources for additional patient enrollment.

\section{Conclusions and Future Directions}

In our work, we propose a novel nonparametric sequential test for composite hypotheses over means of multiple data streams. In the single-arm case, our approach provides a robust option among existing sequential tests and anytime valid confidence sequences. In the multi-arm case, our test pools evidence across arms to avoid union bounds, and empirically outperforms other sequential tests both (i) as a stopping criterion for common pure-exploration bandit problems, such as threshold or best-arm identification, and (ii) in terms of runtime compared with leading approaches for sequential testing. Importantly, our approach is compatible with various existing convex optimization software, avoiding the need for grid-search style methods while maintaining the performance of other testing-by-betting methods.


Despite improved
computational tractability, PEAK requires the use of solvers to test hypotheses; existing parametric $e$-processes have minimal computational overhead relative to any testing-by-betting approach, including PEAK. Our theoretical/empirical results suggest that PEAK is most appropriate for problems where (1) strong parametric assumptions are unrealistic, and (2) one wishes to preserve power against both high-variance and low-variance data streams.


\newpage
\section*{Acknowledgements}
Brian Cho was supported by the NDSEG fellowship. 
\section*{Impact Statement}
This paper presents work whose goal is to advance the field of Machine Learning. There are many potential societal consequences of our work, none which we feel must be specifically highlighted here.

\bibliography{citation}
\appendix

\thispagestyle{empty}

\onecolumn 
\renewcommand\theHtable{Appendix.\thetable}
\counterwithin{table}{section}
\renewcommand\theHfigure{Appendix.\thefigure}
\counterwithin{figure}{section}

\section*{Appendix}
For all $\alpha$-level type I error correctness proofs, we will use the folllowing inequality:
\begin{lemma}[Ville's Maximal Inequality]\label{lem:ville}
    For any non-negative martinagale $L_t$ and any $x>1$, define a potentially infinite stopping time $N \coloneqq \inf\{t \geq 1: L_t \geq x\}$. Then, 
    $$\PP(\exists t: L_t \geq x) \leq \EE[L_0]/x. $$
\end{lemma}

\begin{proof}[Proof of Ville's Maximal Inequality.]
    Define the expected overshoot when $L_t$ surpasses $x$ as $o = \EE[L_N/x | N < \infty] \geq 1$. Using the optional stopping theorem for supermartingales, 
    \begin{align*}
        \EE[L_0] &\geq \EE[L_N]\\
        &= \EE[L_N|N<\infty]\PP(N<\infty) + \EE[L_\infty| N = \infty]\PP(N=\infty)\\
        &\geq \EE[L_N | N < \infty] \PP(N < \infty)\\
        &= ox\PP(N \leq \infty)
    \end{align*}
    Thus, $\PP(N \leq \infty) = \PP(\exists t: L_t \geq x) \leq \frac{\EE[L_0]}{ox} \leq \EE[L_0]/x$. Setting $x = 1/\alpha$, we obtain the desired results for our test martingale. 
\end{proof}

\section{Proofs for Section 3}
\subsection{Proof of Theorem \ref{thm:type_1_error_correctness}}\label{app:proof_thm1_correct}
\begin{proof}[Proof of Theorem \ref{thm:type_1_error_correctness}.]
Note that we are in the setting where $\mu = m$, i.e. our null hypothesis is correct. We first show that $K_t(m)$ is nonnegative for all $t \in \NN$, for $c \geq 1/4$. To prove that $K_t(m) = \prod_{i=1}^t (1+ \frac{(\hat\mu_{i-1} -m)}{c}(X_i - m))$ is nonnegative, it suffices to show that $\frac{(\hat\mu_{i-1} -m)}{c} \in (-1/(1-m), 1/m)$ for all $i \in \NN$. 
    For the upper bound, 
    \begin{align*}
        1/m - \frac{(\hat\mu_{i-1} -m)}{c} &= \frac{1}{m}-\frac{\hat\mu_{i-1}-m}{c} = \frac{1/4 + \hat\sigma_{i-1}^2 - m\hat\mu_{i-1}+m_a^2}{c}\\
        &= \left((m-\frac{1}{2}\hat\mu_{i-1})^2 + c-\frac{1}{4}(\hat\mu_{i-1}^2)  \right) / (cm)\\
        &\geq 0.
    \end{align*}
    For the lower bound,
    \begin{align*}
        \frac{(\hat\mu_{i-1} -m)}{c} - (-\frac{1}{1-m}) &= \frac{(\hat\mu_{i-1}-m)(1-m)+c}{c(1-m_a)}\\
        &= \left(  \hat\mu_{i-1}(1-m) + (m - 1/2)^2 + c-1/4 \right)/(c(1-m)) \\
        &\geq 0.
    \end{align*}
Not only is $K_t(m)$ nonnegative, but it is also a martingale:
$$\EE[K_t(m)| \Fcal_{t-1}] = K_{t-1}(m) * (1+\frac{\hat\mu_{i-1}-m}{c}\EE[(X_i - m)|\Fcal_{t-1}]) = K_{t-1}(m).$$
Thus, by Ville's Maximal Inequality, Lemma \eqref{lem:ville}, $\PP(\exists t: K_t(m) \geq 1/\alpha) = \alpha$.
\end{proof}

\newpage 
\subsection{Proof of Lemma \ref{lem:expection_orcl}.}\label{app:proof_lemma_3}
\begin{proof}[Proof of Lemma \ref{lem:expection_orcl}]
    This proof proceeds in two parts:
    \begin{enumerate}
        \item First, we prove that $\EE[\log(1+\frac{(\mu - m)(X - m)}{c})] \geq 0$, with equality iff $m = \mu$ for $X \sim \text{Bern}(\mu)$.
        \item Second, we show that for any distribution with support $[0,1]$ and mean $\mu$, the expectation $\EE[\log(1+\frac{(\mu - m)(X_i - m)}{c})]$ is minimized at the Bernoulli distribution $\text{Bern}(\mu)$.
    \end{enumerate}

    \paragraph{Step 1:} Under a Bernoulli distribution $X \sim \text{Bern}(\mu)$, the expectation term becomes the following:
    \begin{align*}
        \EE[\log(1+ \frac{(\mu - m)(X - m)}{c}) ] &= \mu_a \log(1+ \frac{(\mu - m_a)(1 - m)}{c}) + (1-\mu_a)\log(1+ \frac{(\mu - m)( - m)}{c}).
    \end{align*}
    Taking the first-order derivative with respect to $m$,
    \begin{align*}
        &\frac{\delta}{\delta m} \EE[\log(1+ \frac{(\mu - m)(X - m)}{c}) ] = \mu_a\frac{2m - 1 -\mu}{c+ (\mu - m )(1-m)} + (1-\mu)\frac{2m - \mu}{c + (\mu - m)(-m)}\\
        &= \frac{\mu (2m - 1 - \mu)(c + (\mu - m)(-m)) + (1-\mu)(2m - \mu)(c + (\mu- m)(1-m))}{(c + (\mu - m)(-m))(c+ (\mu - m )(1-m))}
    \end{align*}
    Note that the denominator is strictly positive, and thus we only need to solve the numerator expression. The numerator expression reduces to the following expression, whose sign only depends on $m - \mu$:
    \begin{align*}
        &\mu (2m - 1 - \mu)(c + (\mu - m)(-m)) + (1-\mu)(2m - \mu)(c + (\mu- m)(1-m))\\
        &=2\left(m - \mu \right)\left(c + (\mu - m)(-m)\right)  + (1-\mu)(2m - \mu)(\mu - m)\\
        &= 2 (m - \mu)\left(c + (\mu - m)(-m) - \frac{1}{2}(1-\mu)(2m - \mu) \right)\\
        &= 2 (m - \mu)\underbrace{(c + m^2 - m + 1/2(\mu - \mu^2) )}_{(i): > 0}
    \end{align*}
    where statement $(i)$ holds due to $m^2 - m \geq -1/4$, $m, \mu \in [0,1]$ and $c \geq 1/4$. This derivative expression is solved when $m = \mu$, which results in $\EE[\log(1+\frac{(\mu-m)(X-m)}{c})] = 0$. The derivative is strictly negative when $m < \mu$, i.e. our function strictly decreases in the region $m \in [0, \mu)$. Symmetrically, the derivative is strictly positive when $m > \mu$, i.e. our function strictly increases in the region $m_a \in (\mu, 1]$ - thus the asymptotic growth rate $G(c,m,P)$ monotonically increases with $|\mu -m|$. 
    This also indicates that in the interior, $m = \mu$ is the minimum. To check endpoints of our feasible region, we only need to consider the cases $m_a \in \{0,1\}$, where both cases give us $\EE[\log(1+\frac{(\mu - m)(X-m)}{1/4 + \sigma^2})]  > 0$. Thus, the global minimum is achieved at $m = \mu$, with $\EE[\log(1+\frac{(\mu - m)(X-m)}{c})] = 0$, and all other values of $m_a$ result in $\EE[\log(1+\frac{(\mu - m)(X-m)}{c})] > 0$.

    \paragraph{Step 2:} Now, we show that for any distribution $P \in \Pcal(\mu)$, the expectation is lower bounded by the Bernoulli distribution $\text{Bern}(\mu)$. Let $X \sim P$, and let $U$ be an independent $\text{Unif}[0,1]$ random variable. Note that the random variable $R = \mathbf{1}[U \leq X]$ has the following properties:
    $$\EE_U[\mathbf{1}[U \leq X]|X] = X, \quad \EE[\mathbf{1}[U \leq X]] = \mu.   $$
    Thus, $R$ is a Bernoulli random variable, with mean $\mu$. Our function $f(X) = \log(1+\frac{(\mu - m)(X_i - m)}{c})$ is a concave function for all $\mu, m \in [0,1]$, and by using the conditional version of Jensen's inequality:
    $$\EE[f(R)|X] \leq f(\EE[R|X]) = f(X). $$
    By taking the expectation over $X$, we obtain the desired inequality:
    $$\EE[f(R)] \leq \EE[f(X)] \implies \EE_{X \sim \text{Bern}(\mu)}[\log(1+ \frac{(\mu - m)(X_i - m)}{c})] \leq \EE_{P}[\log(1+ \frac{(\mu - m)(X_i - m)}{c})]. $$
\end{proof}

\newpage
\subsection{Proof of Theorem \ref{thm:plug_in_process_power_1}.} \label{app:proof_thm_2}
\begin{proof}[Proof of Theorem \ref{thm:plug_in_process_power_1}.]
We first prove the statement that $K_t(m) \rightarrow \infty$ almost surely as $t \rightarrow \infty$. First, note that we can include indicator functions to separate the positive and negative terms within the following expectation: 
    \begin{align*}
        &\EE[\log((1+\frac{\mu-m)(X_i-m)}{c}]=\\
        &\EE\left[\mathbf{1}[X_i \geq m]\log((1+\frac{\mu-m)(X_i-m)}{c}\right] + \EE\left[\mathbf{1}[X_i < m_a]\log((1+\frac{\mu-m)(X_i-m)}{c})\right] > 0.
    \end{align*}
    Consider the following function, which equals the summation above at $\epsilon = 0$:
    $$f(\epsilon, m) = \EE\left[\mathbf{1}[X_i \geq m]\log(1+\frac{(\mu-\epsilon-m)(X_i-m_a)}{c})\right] + \EE\left[\mathbf{1}[X_i \geq m_a]\log(1+\frac{(\mu+\epsilon-m)(X_i-m_a)}{c})\right] $$
    Note that $f(\epsilon, m) = \EE[\log((1+\frac{\mu-m_)(X_i-m)}{c}] > 0$ at $\epsilon = 0$ (proven above). Because $f(\epsilon, m)$ is continuous with respect to $\epsilon$, we know there exists an $\epsilon(m)$ such that $f(\epsilon(m), m) =  \EE[\log((1+\frac{(\mu-m)(X_i-m)}{c}]/2> 0$.  
    By the Kolmogorov Strong Law of Large Numbers, we know $t(m, \omega) \in \NN$, $t(m, \omega) < \infty$ such that $\hat\mu_{t} \in [\mu -\epsilon(m), \mu+\epsilon(m)]$ for all $t \geq t(m, \omega)$ almost surely, where $\omega \in \Omega$ s.t. $\PP(\Omega) = 1$. We denote $t(m, \omega)$ as $t(m)$ to keep notational clutter minimal.
    The log capital process $l_t(m)$, for $t > t(m)$, takes the following form: 
    \begin{equation}\label{eq:log_cap_expansion}
        l_t(m) = \underbrace{\sum_{i=1}^{t(m)}\log(1+\frac{(\hat\mu_{i-1}-m_a)(X_i-m_a)}{c})}_{(a)}  +  \underbrace{\sum_{i=t(m)}^t \log(1+\frac{(\hat\mu_{i-1}-m_a)(X_i-m_a)}{c})}_{(b)}.
    \end{equation}
    Summation $(a)$ is guaranteed to be finite almost surely by $t(m) < \infty$, and we will now show that summation $(b)$ diverges to $\infty$ for all $m_a \neq \mu$. We can rewrite summation $(b)$ as the following:
    \begin{align*}
        &\sum_{i=t(m)}^t \log(1+\frac{(\hat\mu_{i-1}-m_a)(X_i-m_a)}{c}) = \\
        &\sum_{i=t(m)}^\infty \mathbf{1}[X_i \geq m_a] \log(1+\frac{(\hat\mu_{i-1}-m_a)(X_i-m_a)}{c}) + \sum_{i=t(m)}^t \mathbf{1}[X_i < m_a] \log(1+\frac{(\hat\mu_{i-1}-m_a)(X_i-m_a)}{c}) 
    \end{align*} 
    Recall that for all $i \geq t(m)$, $\hat\mu_i \in [\mu - \epsilon(m), \mu + \epsilon(m)]$. Using this, we can compare our summation $(b)$ to the strictly smaller process for any realization $\{X_i\}_{i=1}^\infty$:
    \begin{align*}
        (b) \geq &\sum_{i=t(m)}^t \mathbf{1}[X_i \geq m_a] \log(1+\frac{(\mu - \epsilon(m_a)-m_a)(X_i-m_a)}{c})\\
        &+ \sum_{i=t(m)}^t \mathbf{1}[X_i < m_a] \log(1+\frac{(\mu + \epsilon(m_a)-m_a)(X_i-m_a)}{c}). 
    \end{align*}
    
    $$$$
    By another application of SLLN, note that the RHS converges a.s. to a positive constant as $t\rightarrow\infty$:
    \begin{align*}
    &\frac{\sum_{i=t(m)}^t \mathbf{1}[X_i \geq m_a] \log(1+\frac{(\mu - \epsilon(m_a)-m_a)(X_i-m_a)}{c})}{t-t(m)} + \frac{\sum_{i=t(m)}^t \mathbf{1}[X_i \geq m_a] \log(1+\frac{(\mu - \epsilon(m_a)-m_a)(X_i-m_a)}{c})}{t-t(m)} \\
    &\rightarrow f(\epsilon(m), m)  > 0 
    \end{align*}
    which implies that the non-normalized summation diverges to $\infty$:
    $$ \mathbf{1}[X_i \geq m_a] \log(1+\frac{(\mu - \epsilon(m_a)-m_a)(X_i-m_a)}{c}) + \sum_{i=t(m)}^t \mathbf{1}[X_i < m_a] \log(1+\frac{(\mu + \epsilon(m_a)-m_a)(X_i-m_a)}{c}) \rightarrow \infty.$$
    Because $(b)$ is at least as large as the series which diverges to $\infty$ a.s., summation $(b) \rightarrow \infty$ a.s. as well, resulting in the following limit for $K_t(m)$:
    $$l_t(m) = \underbrace{(a)}_{\text{finite.}} + \underbrace{(b)}_{\rightarrow \infty} \rightarrow \infty, \quad\quad K_t(m) = \exp(l_t(m)) \rightarrow \infty. $$
    We now prove that the stopping time (i.e. time until rejection) of the test $T_t(m, \alpha)$ is finite for any fixed $m \neq \mu$ and any $\alpha \in (0,1)$. Using $\epsilon(m), t(m)$ as above, if $T_t(m, \alpha) = 1$ for some $t \leq t(m)$, then we are guaranteed that the stopping time $\tau = \{\min_{t \in \NN} t: T_t(m,\alpha) = 1 \}$ is finite. If not, then we rewrite the test as follows using our log-capital expression $l_t(m)$ in Equation \eqref{eq:log_cap_expansion}: 
    $$T_{t}(m, \alpha) = \mathbf{1}[i \leq t : \sum_{i=t(m)}^t \log(1+\frac{(\hat\mu_{i-1}-m_a)(X_i-m_a)}{c}) \geq \ln(1/\alpha) - \sum_{i=1}^{t(m)}\log(1+\frac{(\hat\mu_{i-1}-m_a)(X_i-m_a)}{c})]. $$
    Now, by replacing the term $\sum_{i=t(m)}^t \log(1+\frac{(\hat\mu_{i-1}-m_a)(X_i-m_a)}{c})$ with a strictly smaller term $\sum_{i=t(m)}^t \mathbf{1}[X_i \geq m_a] \log(1+\frac{(\mu - \epsilon(m_a)-m_a)(X_i-m_a)}{c}) + \sum_{i=t(m)}^t \mathbf{1}[X_i < m_a] \log(1+\frac{(\mu + \epsilon(m_a)-m_a)(X_i-m_a)}{c})$, we construct a new stopping rule, $T'_t(m,a)$, which substitutes the LHS of the inequality for a strictly smaller term (shown above) as well as normalizing by $t-t(m)$.
    \begin{align*}
        T_{t}'(m, \alpha) = \mathbf{1}[i \leq t :\ & \ \frac{1}{t-t(m)}\sum_{i=t(m)}^t \mathbf{1}[X_i \geq m_a] \log(1+\frac{(\mu - \epsilon(m_a)-m_a)(X_i-m_a)}{c}) + \\
        & \ \frac{1}{t-t(m)}\sum_{i=t(m)}^t \mathbf{1}[X_i < m_a] \log(1+\frac{(\mu + \epsilon(m_a)-m_a)(X_i-m_a)}{c})\\ 
        & \ \geq \frac{1}{t-t(m)}\left(\ln(1/\alpha) - \sum_{i=1}^{t(m)}\log(1+\frac{(\hat\mu_{i-1}-m_a)(X_i-m_a)}{c})\right)].
    \end{align*}
    By the SLLN and by the fact that the LHS converges to $f(\epsilon(m), m)/2 > 0$, there exists a finite time $t'(m, \omega) \in \NN$ almost surely for each sample path $\omega \in \Omega$ such that for all $t \geq t'(m, \omega) = t'(m)$, the following holds:
    \begin{align*}
        &\frac{\sum_{i=t(m)}^t \mathbf{1}[X_i \geq m_a] \log(1+\frac{(\mu - \epsilon(m_a)-m_a)(X_i-m_a)}{c})}{t-t(m)} + \frac{\sum_{i=t(m)}^t \mathbf{1}[X_i \geq m_a] \log(1+\frac{(\mu - \epsilon(m_a)-m_a)(X_i-m_a)}{c})}{t-t(m)} \\
        & \geq f(\epsilon(m), m)/4 > 0. 
    \end{align*}
    Let $t''(m)$ be the (random) time, $t''(m) \geq t'(m)$ such that the following holds:
    $$ \frac{1}{t''(m)-t(m)}\left(\ln(1/\alpha) - \sum_{i=1}^{t(m)}\log(1+\frac{(\hat\mu_{i-1}-m_a)(X_i-m_a)}{c})\right)] \leq f(\epsilon(m), m)/4.$$
    Note that $t''(m)$ must also be finite, as $(\ln(1/\alpha) - \sum_{i=1}^{t(m)}\log(1+\frac{(\hat\mu_{i-1}-m_a)(X_i-m_a)}{c})$ is finite due to the finiteness of $t(m)$. Thus, $\PP(t < \infty: T'_t(m, \alpha) = 1) = 1$. Because the (unnormalized ) LHS of the inequality is strictly smaller for $T'_t(m, \alpha)$ than for $T_t(m ,\alpha)$, if $T_t'(m, \alpha) = 1$, then $T_t(m, \alpha) = 1$ as well. Thus,
    $$\PP(t < \infty : T_t(m,\alpha) = 1) \geq \PP(t < \infty: T'_t(m, \alpha) = 1) = 1 \implies\PP(t < \infty : T_t(m,\alpha) = 1)=1.  $$

\end{proof}

\newpage


\subsection{Proof of Lemmas \ref{lem:c_small_optimal} and and \ref{lem:orcl_growth_rate} }\label{app:proof_lem_2}
\begin{proof}[Proof of Lemma \ref{lem:c_small_optimal} .]
Let $c, c' \in [1/4, \infty)$, with $c' < c$. Let $P$ be any distribution in $\Pcal(m)$. Taking the difference of the two asymptotic growth rates, we obtain:
    \begin{align*}
        &\EE_P[\log(1+\frac{(\mu-m)(X - m)}{c'})] - \EE_P[\log(1+\frac{(\mu-m)(X - m)}{c})] \\
        &= \EE_P[\log(c'+(\mu-m)(X - m)) - \log(c')] - \EE_P[\log(c+(\mu-m)(X - m)) - \log(c)]\\
        &= \EE_P[\log(\frac{c' + (\mu - m)(X-m)}{c + (\mu-m)(X - m)})] + \log(\frac{c}{c'})\\
        &= \EE_P[\log(c' + (\mu - m)(X-m))]-\EE_P[\log(c' + (\mu - m)(X-m))] + \log(c) - \log(c').
    \end{align*}
    Then, we use the identity: $\log(x+y) = \log(x) + \log(y) + \log(\frac{1}{x} + \frac{1}{y})$ to obtain the following expressions for each expectation term:
    \begin{align*}
        \EE_P[\log(c' + (\mu - m)(X-m))] &= \log(c')+\EE_P[\log(\mu - m)(X-m) + \log(\frac{1}{c'} + \frac{1}{(\mu - m)(X-m)})] \\
        \EE_P[\log(c + (\mu - m)(X-m))] &=  \log(c)+\EE_P[\log(\mu - m)(X-m) + \log(\frac{1}{c} + \frac{1}{(\mu - m)(X-m)})]
    \end{align*}
    Returning to the main expression,
    \begin{align*}
        &\EE_P[\log(1+\frac{(\mu-m)(X - m)}{c'})] - \EE_P[\log(1+\frac{(\mu-m)(X - m)}{c})] \\
        &= \EE_{P}[\log(1/c' + (\mu-m)^{-1}(X-m)^{-1}) - \log(1/c + (\mu-m)^{-1}(X-m)^{-1})] \\
        &> 0.
    \end{align*}
    where the last inequality is because the LHS of the expression within the expectation is always strictly greater than the right hand side of the expression ($1/c'>1/c$) for all $X$ in the support of P.
\end{proof}

\begin{proof}[Proof of Lemma \ref{lem:orcl_growth_rate}]\label{app_proof_lem_4}
    First, note that the maximal rate of growth for any martingale-based test for the Bernoulli setting is given by the KL-divergence \cite{pérezortiz2023estatistics}:
    $$G^{\text{opt.}}(m, \mu) = \EE_{\mu}[\frac{\mu}{m}] = \mu\log(\frac{\mu}{m}) + (1-\mu)\log(\frac{1-\mu}{1-m}).  $$
    In the Bernoulli case, the expectation of $G(c,m,P)$ is given by:
    $$G(c,m,P) = \mu\log(1+\frac{(\mu-m)(1-m)}{c}) +  (1-\mu)\log(1+\frac{(\mu-m)(-m)}{c}).$$
    Taking the ratio, we obtain the desired result:
    \begin{align*}
        \frac{G(c,m,P)}{G^{\text{opt.}}(m, \mu)} &= \frac{\log\left((1+\frac{(\mu-m)(1-m)}{c})^\mu (1+\frac{(\mu-m)(-m)}{c})^{1-\mu}\right)}{\log\left((\frac{\mu}{m})^\mu (\frac{1-\mu}{1-m})^{1-\mu} \right)}\\
        &= \log\left((1+\frac{(\mu-m)(1-m)}{c})^\mu (1+\frac{(\mu-m)(-m)}{c})^{1-\mu} -(\frac{\mu}{m})^\mu (\frac{1-\mu}{1-m})^{1-\mu}\right)
    \end{align*}
\end{proof}

\subsection{Proof of Proposition 1}\label{app:proof_prop_1}
\begin{proof}[Proof of Proposition \ref{prop:convex}]
    We prove this theorem for the single arm case $K_t(a) = \sum_{i=1}^t \left(1+\frac{(\hat\mu_{i-1} - m)(X_i - m)}{c} \right)$ - note that this theorem can be generalized to the multi-armed case by taking partial derivatives, which reduces into the single arm case for $K_t^a(a) = \sum_{i=1}^t \mathbf{1}[A_t=a]\left(1+\frac{(\hat\mu_{i-1}(a)-m_a)(X_i - m_a)}{c}\right)$. We prove this statement through induction, under the assumption that $K_t(m)$ is strictly convex for any $m \in \RR$, i.e. its derivative is strictly monotone increasing for any realization of $(X_i)_{i=1}^t$. 

    As the base case for $t=1$, for any $X_1 \in [0,1]$, (1) $K_t(m)$ is clearly convex for any $m \in \RR$, (2) $\frac{\delta}{\delta m}K_t(m)$ is strictly monotone increasing for $m \in \RR$, and (3) there exists a unique minimizing value $m^*_1$ between $[0,1]$ that satisfies $\sum_{i=1}^t \frac{2m - \hat\mu_0 - X_1}{c + (\hat\mu_0 - m)(X_i - m)} = 0$ and $K_t(m^*_{t}) \geq 0$. 
    $$K_t(m) = \frac{c + (\hat\mu_{0} - m)(X_1 - m)}{c}, \quad  \frac{\delta}{\delta m}K_t(m) = (2m - \hat\mu_0 - X_1)/c, \quad m^*_t = \frac{\hat\mu_0 + x_1}{2} \in [0,1].  $$
    We now turn to the induction step. By our inductive hypothesis, we assume that (1) $K_{t-1}(m)$ is strictly convex for $m \in \RR$, (2) $\frac{\delta}{\delta m} K_{t-1}(m)$ is strictly monotonically increasing for $m \in \RR$, and (3) the minimizing value $m^*_{t-1}$ is unique, between $[0,1]$, and $K_{t-1}(m^*_{t-1})\geq 0$. We first rewrite the derivative of $K_t(m)$ as the following:
    \begin{align*}
        \frac{\delta}{\delta m} K_t(m) &= \sum_{i=1}^t (2m - \hat\mu_{i-1} - X_i) \left[ \prod_{j\neq i}^t c + (\hat\mu_{i-1} - m)(X_i - m)  \right]\\
        &= (c + (\hat\mu_{i-1} - m)(X_i - m)) \sum_{i=1}^{t-1} (2m - \hat\mu_{i-1} - X_i) \left[ \prod_{j\neq i}^{t-1} c + (\hat\mu_{i-1} - m)(X_i - m)  \right] \\
        &\quad + (2m - \hat\mu_{t-1} - X_i) \left[\prod_{i=1}^{t-1} c+ (\hat\mu_{i-1}-m)(X_i - m)\right]
    \end{align*} 
    To show this is monotone, we use the sign of the second derivative. 
    \begin{align*}
        \frac{\delta^2}{\delta m^2} K_t(m) &= 2 \left[\prod_{i=1}^{t-1} c+ (\hat\mu_{i-1}-m)(X_i - m)\right] 
        \\ & \quad +  2[(2m - \hat\mu_{t-1} - X_t)]\sum_{i=1}^{t-1} (2m - \hat\mu_{i-1} - X_i) \left[ \prod_{j\neq i}^{t-1} c + (\hat\mu_{i-1} - m)(X_i - m)  \right]\\
        & \quad + (c + (\hat\mu_{i-1} - m)(X_i - m)) \frac{\delta}{\delta m}\sum_{i=1}^{t-1} (2m - \hat\mu_{i-1} - X_i) \left[ \prod_{j\neq i}^{t-1} c + (\hat\mu_{i-1} - m)(X_i - m)  \right]\\
        &= 2c^{t-1} K_{t-1}(m) + 2c^{t-1}[(2m - \hat\mu_{t-1} - X_t)] \frac{\delta}{\delta m} K_{t-1}(m) + c^{t-1}(c + (\hat\mu_{i-1} - m)(X_i - m)) \frac{\delta^2}{\delta m^2} K_{t-1}(m).
    \end{align*}
    Using the fact that $c^{t-1}(c + (\hat\mu_{i-1} - m)(X_i - m)) \frac{\delta^2}{\delta m^2} K_{t-1}(m) > 0$ for all $m \in [0,1]$, we obtain the following inequality for the second derivative expression:
        \begin{align*}
            \frac{\delta^2}{\delta m^2}K_t(m) > 2c^{t-1}(K_{t-1}(m) + (2m - \hat\mu_{t-1} - X_t)\frac{\delta}{\delta m} K_{t-1}(m) )
        \end{align*}
    By our inductive hypothesis, (i) for any $\Tilde{m} \in \RR$, $K_{t-1}(\Tilde{m}) \geq K_{t-1}(m^*_{t-1}) \geq 0$, and (ii) by convexity $K_{t-1}(m) \geq K_{t-1}(\Tilde{m}) + (m - \Tilde{m})\frac{\delta}{\delta m}K_{t-1}(m)$. Using these assumptions, we obtain another lower bound on the second derivative:
        \begin{align*}
            \frac{\delta^2}{\delta m^2}K_t(m) > 2c^{t-1}(K_{t-1}(\Tilde{m}) + (m - \Tilde{m}+ 2m - \hat\mu_{t-1} - X_t)\frac{\delta}{\delta m} K_{t-1}(m) ), \quad \forall \Tilde{m} \in \RR.
        \end{align*}
    Because $\Tilde{m}$ is any real number, we can pick $\Tilde{m} = (3m -\hat\mu_{t-1} - X_t)$ to obtain that $\frac{\delta^2}{\delta m^2} K_t(m) > 2c^{t-1}K_{t-1}(\Tilde{m})\geq 0$. Thus, for all $m \in \RR$, $\frac{\delta^2}{\delta m^2} K_t(m) > 0$, and so (1) $\frac{\delta}{\delta m}K_t(m)$ is strictly monotonically increasing and (2) $K_t(m)$ is strictly convex for all $m \in \RR$. By the monotonocity of the derivative, there exists only one $m_t^* \in \RR$ where $\frac{\delta}{\delta m}K_t(m) = 0$, which must be a minimum by convexity. To show that $m^*_t \in [0,1]$, we consider the derivative function again:
    \begin{align*}
        \frac{\delta}{\delta m} K_t(m) &= \sum_{i=1}^t (2m - \hat\mu_{i-1} - X_i) \left[ \prod_{j\neq i}^t c + (\hat\mu_{i-1} - m)(X_i - m)  \right]\\
        &= (c + (\hat\mu_{i-1} - m)(X_i - m)) \sum_{i=1}^{t-1} (2m - \hat\mu_{i-1} - X_i) \left[ \prod_{j\neq i}^{t-1} c + (\hat\mu_{i-1} - m)(X_i - m)  \right] \\
        &\quad + (2m - \hat\mu_{t-1} - X_i) \left[\prod_{i=1}^{t-1} c+ (\hat\mu_{i-1}-m)(X_i - m)\right]\\
        &= c^{t-1}(c + (\hat\mu_{i-1} - m)(X_i - m))\frac{\delta}{\delta m} K_{t-1}(m) + c^{t-1}(2m - \hat\mu_{t-1} - X_i)K_{t-1}(m)
    \end{align*} 
    We claim that $\min(m_{t-1}^*, \frac{\hat\mu_{t-1} + X_t}{2}) \leq m_{t}^* \leq \max(m_{t-1}^*, \frac{\hat\mu_{t-1} + X_t}{2})$. Assume this to be false: then, one of the following must be true.

    \begin{itemize}
        \item Case 1: $m < \min(m_{t-1}^*, \frac{\hat\mu_{t-1} + X_t}{2}).$
        
        In this case, $(2m - \hat\mu_{t-1}-X_t) < 0$. Because $m_{t-1}^*$ is the unique point where $\frac{\delta}{\delta m} K_{t-1}(m) = 0$, and $\frac{\delta}{\delta m} K_{t-1}(m)$ is monotonic, $\frac{\delta}{\delta m} K_{t-1}(m) < 0$. Note that $c^{t-1}(c + (\hat\mu_{i-1} - m)(X_i - m)) > 0$, and $c^{t-1}K_{t-1}(m) > 0$, and thus $\frac{\delta}{\delta m}K_t(m) < 0$. 

        \item Case 2: $m > \max(m_{t-1}^*, \frac{\hat\mu_{t-1} + X_t}{2}).$
        
         In this case, $(2m - \hat\mu_{t-1}-X_t) > 0$. Because $m_{t-1}^*$ is the unique point where $\frac{\delta}{\delta m} K_{t-1}(m) = 0$, and $\frac{\delta}{\delta m} K_{t-1}(m)$ is monotonic, $\frac{\delta}{\delta m} K_{t-1}(m) > 0$. Note that $c^{t-1}(c + (\hat\mu_{i-1} - m)(X_i - m)) > 0$, and $c^{t-1}K_{t-1}(m) > 0$, and thus $\frac{\delta}{\delta m}K_t(m) > 0$. 
    \end{itemize}

    Neither case can be true, and therefore $m_{t}^* \in [\min(m_{t-1}^*, \frac{\hat\mu_{t-1} + X_t}{2}), \max(m_{t-1}^*, \frac{\hat\mu_{t-1} + X_t}{2})]$, which is a subset of $[0,1]$ by the assumption that $m^*_{t-1} \in [0,1]$ and $\hat\mu_{t-1} + X_t \in [0,2]$. For any $m \in [0,1]$, $K_t(m) > 0$ by definition, and so (3) $m_t^*$, the unique minimizer of $K_t(m)$, lies within $[0,1]$ and $K_t(m_t^*) \geq 0$. This confirms our inductive hypothesis, and so for all $t \in \NN$, $K_t(m)$ is a convex function with a unique minimizer $m_t^*$ satisfying $\frac{\delta}{\delta m} K_t(m_t^*) = 0$.

\end{proof}

\section*{Proofs for Section 4}

\subsection{Proof of Theorem \ref{thm:product_theorem_correctness}}\label{app:proof_thm_3}
\begin{proof}[Proof of Theorem \ref{thm:product_theorem_correctness}]
We first note that because $E_t(m) = \frac{1}{W}\sum_{a\in [W]}K_t^a(m_a)$, and $K_t^a(m_a)\geq 0$ for all $a \in [A]$, $t \in \NN$, $E_t(m)$ is nonnegative. We now show that this process is a martingale for $m=\mu$, regardless of sampling policy:
\begin{align*}
    \EE_t[E_t(m)|\Fcal_{t-1}] &= \sum_{a \in [W]}\pi_t(a)\left[ (1-\pi_t(a)) + \pi_t(a)*(1-0) \right] *E_{t-1}(m)\\
    &= E_{t-1}(m).
\end{align*}
Thus, by Ville's inequality \ref{lem:ville}, we obtain that the crossing probability $\PP(\exists t \in \NN: T_t(m) = \mathbf{1}[E_t(m)\geq 1/\alpha] = 1) \leq 1/\alpha$, guaranteeing our desired Type I error coverage. 
\end{proof}

\subsection{Proof of Theorem \ref{thm:product_capital_diverges}}\label{app:_proof_thm_4.}
\begin{proof}[Proof of Theorem \ref{thm:product_capital_diverges}]\label{app:proof_thm_4}
    We prove this statement by building upon the proof of Theorem \ref{thm:plug_in_process_power_1} and using Lemma \ref{lem:convergence_bandit} below:
    \begin{lemma}[Fact E.1 of \citealp{fact_e1}]\label{lem:convergence_bandit}
        Suppose that $Y_n \rightarrow Y$ a.s. as $n \rightarrow \infty$, and $N(t) \rightarrow \infty$ a.s. as $t \rightarrow \infty$. Then $Y_{N(t)} \rightarrow Y$ a.s. as $t \rightarrow \infty$.
    \end{lemma}
    This establishes necessary convergence results for the assumptions we place on our policy, i.e. that for our sampling scheme $\pi$, $N_t(a) \rightarrow \infty$ a.s. for some $a \in \Ical(m)$. Let $\omega \in \Omega$ denote an event / sample path, where $\PP(\Omega) = 1$. Because there exists an arm $a^*(\omega)\in \Ical(\omega)$ such that $N_t(a^*(\omega))(\omega) \rightarrow \infty$ a.s., $K_t^{a^*(\omega)}(m_{a^*(\omega)}) \rightarrow \infty$ a.s. by Theorem \ref{thm:plug_in_process_power_1}, and $K_t^a(m_a) \geq 0$ for all $a \neq a^*(\omega)$, $E_t(m)(\omega) \rightarrow \infty$ almost surely. By the same argument as Theorem \ref{thm:plug_in_process_power_1}, this implies that the stopping time for test $T_t(m, \alpha)$ is finite for any fixed $m \neq \mu$, $\alpha \in (0,1)$.

    \subsection{Proof of Lemma \ref{lem:union_bad}}\label{app:proof_lem_5}

    This follows from the fact that for all $t \in \NN$, $(A_i,X_i)_{i=1}^t \in  ([W]\times [0,1])^t$, $K_t^a(m_a)$ is nonnegative. Thus, 
    $$\frac{1}{W}\sum_{a \in [W]} K_t^a(a) \geq \frac{1}{W} K_t^a(m_a). $$
    Note that $T'_t(m, \alpha) = \prod_{a\in [W]}\mathbf{1}[ \max_{1 \leq i \leq t} K_i^a(m_a) \geq  W/\alpha] = 1$, then there exists a time $j \in [1,...,t]$ such that $K_t^a(m_a) \geq W/\alpha$ for some $a$. Then, for the same time $j \leq t$, $\sum_{a \in [W]}K_j^a(a) \geq W/\alpha$, which gives the desired result:
    $$T_i(m, \alpha) =  \mathbf{1}\left[\max_{1\leq i \leq t} \sum_{a \in [W]} K_t^a(a) \geq  \frac{W}{\alpha}\right] = 1. $$
    .

    \subsection{Proof of Proposition \ref{prop:BAI_min}}\label{app:proof_prop_BAI_min}

    At time $t$, denote the global minimizer of $E_t(m)$ solution as $\Tilde{m}$, which is given entry-wise by the following condition:
    $$\frac{\delta}{\delta m_a} E_t(m) = 0  \iff \sum_{i=1}^t \mathbf{1}[A_t = a] \frac{2\Tilde{m}_a - X_i - \hat\mu_{i-1}(a)}{c + (\hat\mu_{i-1}(a))}   = 0 \quad \forall a \in [W].$$  
    Recall that our proposed solution is given by the form, with respect to partition $R_{\text{BAI}}(a')$, is given by the following form:
    \begin{equation*}
    m^*_a = \begin{cases}
        \tilde{m}_a \quad &\text{if } \Tilde{m}_a < q\\
        q \quad & \text{if } \Tilde{m}_a \geq q
    \end{cases}
    \end{equation*}
    where $q \in [0,1]$ is a constant that satisfies $$\sum_{a \in W: \Tilde{m}_a \geq q} \gamma(m^*,a) \sum_{i=1}^t  \mathbf{1}[A_t = a] \frac{2q - X_i - \hat\mu_{i-1}(a)}{c + (\hat\mu_{i-1}(a)-q)(X_i - q)} = 0,$$
    where $\gamma(m,a) = \prod_{i=1}^t \left( \mathbf{1}[A_i = a](1 + (X_i - m_a)(\hat\mu_{i-1}(a) - m_a)/c) + (1-\mathbf{1}[A_i = a]) \right)$.
    
    We verify that this is solution is indeed optimal through the KKT conditions (Chapt. 5, \citealp{boyd2004convex}). To see this, we first convert our problem to standard convex optimization notation:
    \begin{align*}
        &\min_{m} \frac{1}{W}\sum_{a \in [W]} K_t^a(m_a) \\
        & \text{s.t. } \mu_a - \mu_{a'} \leq 0 \quad \forall i \in [W]\setminus\{a'\}.
    \end{align*}
    The lagrangian dual of this problem takes the form, for dual variable $\lambda \in \RR^{W-1}$:
    $$ \Lcal(m, \lambda) = \frac{1}{W}\sum_{a \in [W]} K_t^a(m_a) + \sum_{a \neq a'}\lambda_a(m_a - m_{a'}) \quad \text{s.t. } \lambda \geq 0 . $$
    Our proposed solution in terms of the Lagrangian is as follows:
    $$\lambda_a = \begin{cases}
        0  &\text{if } \Tilde{m}_a < q \\
        -\frac{1}{W}\frac{\delta}{\delta m_a} K_t^{a}(q) & \text{if } \Tilde{m}_a \geq q \end{cases}
        \quad  m^*_a = \begin{cases}
        \tilde{m}_a  &\text{if } \Tilde{m}_a < q\\
        q  & \text{if } \Tilde{m}_a \geq q
    \end{cases}.  $$

    Note that $m^*_a \leq q = m^*_{a'}$, which satisfies the (1) primal feasibility requirements. For (2) dual feasibility, note that $\lambda_a = 0$ for $\Tilde{m}_a < q$, and $-\frac{1}{W} \frac{\delta}{\delta m_a} K_t^a(q) > 0$ for $\Tilde{m}_a \geq q$, which satisfies dual feasibility. The (3) complementary slackness condition is also satisfied, i.e.:
    $$\lambda_a(m^*_{a} - m^*_{a'}) = 0 \quad \forall a \in [W]\setminus \{a'\}. $$
    Lastly, to check (4) Lagrangian stationary, we begin with $a$ such that $\Tilde{m}_a < \tilde{m}_{a'}$:
    $$\frac{\delta}{\delta m_a} L(m^*) =  \frac{1}{W} \frac{\delta}{\delta m_a} K_t^a(m_a) = 0.$$
    For $a$ such that $\Tilde{m}_a > \Tilde{m}_{a'}$, 
    $$\frac{\delta}{\delta m_a} L(m^*) = \frac{1}{W} \frac{\delta}{\delta m_a} K_t^a(q) - \frac{\delta}{\delta m_a} K_t^a(q) = 0. $$
    Finally, for $m_{a'}$,
    $$\frac{\delta}{\delta m_{a'}} L(m^*) =  \frac{1}{W} \frac{\delta}{\delta m_a} K_t^{a'}(q)  - \sum_{a: \tilde{m}_a \geq q} \lambda_a =\frac{1}{W} \sum_{a\in [W]: \Tilde{m}_a \geq q} \frac{\delta}{\delta m_a} K_t^a(q).  $$
    Note that this equation is solved by the condition:
    $$ \sum_{a \in [W]: \Tilde{m}_a \geq q}\gamma(m,a)\sum_{i=1}^t \mathbf{1}[A_t = a'] \frac{2m_a - X_i - \hat\mu_{i-1}(a)}{c + (X_i - m_a)(\hat\mu_{i-1}(a') - m_a)} = 0,$$
    where $\gamma(m,a) = \prod_{i=1}^t \left( \mathbf{1}[A_i = a](c + (X_i - m_a)(\hat\mu_{i-1}(a) - m_a)) + (1-\mathbf{1}[A_i = a]) \right)$, and thus our solution solves this equation. Note that Slater's condition (pg. 226, \citealp{boyd2004convex}) is trivially satisfied by $R_{\text{BAI}}(a')$ for any $a' \in [W]$ (i.e. feasible region has at least one interior point), and thus strong duality holds. Then, our KKT conditions are necessary and sufficient conditions for optimality, and therefore our solution solves our primal constrained minimization problem.

\end{proof}

\section{Experiment Details.}\label{app:additional_experiments}
In this section, we provide additional details on the baseline methods of comparison, the data-adaptive sampling schemes used by in Section 5.2, additional details regarding Table \ref{table:runtimes}, and additional simulations using non-standard underlying distributions for each arm. 

\subsection{Additional Details on Baseline Confidence Sequences / Anytime Valid Tests.}

We provide additional details on confidence sequences used in the single arm experiments in Section \ref{sec:experiments}.

\paragraph{Sub-Bernoulli.} We use the package \texttt{confseq} \cite{confseq}, and use the function $\texttt{bernoulli\_confidence\_interval}$ with parameters optimal intrinsic time equal to $t=2500$, and $\alpha = \alpha_{\text{opt}} = 0.05$.

\paragraph{PrPlH.} The confidence sequence for PrPlH \cite{waudbysmith2022estimating} is given by:
$$C_t^{\text{PrPlH}} = \left(\frac{\sum_{i=1}^t \lambda_i X_i}{\sum_{i=1}^t \lambda_i}  \pm \frac{\log(2/\alpha) + \sum_{i=1}^t\lambda_i^2/8}{\sum_{i=1}^t \lambda_i}\right), \quad \lambda_i = \min(\sqrt{\frac{8\log(2/\alpha)}{t\log(t+1)}}, 1). $$

\paragraph{Empirical Bernstein.} The confidence sequence for Emp. Bern. \cite{waudbysmith2022estimating} is given by:
$$C_t^{\text{Emp. Bern.}} = \left(\frac{\sum_{i=1}^t \lambda_i X_i}{\sum_{i=1}^t \lambda_i}  \pm \frac{\log(2/\alpha) + \sum_{i=1}^t v_i \psi_E(\lambda_i)}{\sum_{i=1}^t \lambda_i}\right). $$
where $\lambda_t^{\text{PrPl}} = \min(\sqrt{\frac{2\log(2/\alpha)}{\hat\sigma^2{t-1}t\log(t+1)}}, 1/2)$, $\hat\sigma^2 = \frac{1/4 + \sum_{i=1}^t (X_i - \hat\mu_i)^2}{t+1}$, $\hat\mu_t = \frac{1/2 + \sum_{i=1}^t X_i}{t+1}$, $v_t = 4(X_i - \hat\mu_{i-1})^2$, and $\psi_E(\lambda) = (-\log(1-\lambda) - \lambda)/4$. 

\paragraph{Hedged.} We use the empirical-bernstein based variant of the hedged capital process discussed in \citet{waudbysmith2022estimating}:
\begin{align*}
    & \kappa_t^\pm(m) = \max\left(\theta \kappa_t^+(m), (1-\theta)\kappa_t^-(m) \right),\\
    & \kappa_t^{+}(m) = \prod_{i=1}^t(1+ \lambda_i^+(m)(X_i - m)) \\
    & \kappa_t^-(m) = \prod_{i=1}^t (1- \lambda_i^-(m)(X_i - m)),\\
    & C_t^{\text{Hedged}} = \{m \in [0,1]: \kappa_t^{\pm} < 1/\alpha\}.
\end{align*}

with $\theta = 1/2$, 100 grid points for each simulation, and $\lambda_t^+ = \lambda_t^- = \lambda_t^{\text{PrPl}}$ as defined for the Emp. Bern. approach.

\paragraph{Baseline Methods for Multi-Arm Case.}
The termination condition for a single arm $a$ (i.e. below/above threshold $\epsilon$) is determined by the following anytime valid confidence interval in \cite{kano2018good}, which we refer to as $C_t^{\text{base}}$:
$$C_t^{\text{base}}(a) = \left\{ \hat\mu_t(a) \pm \sqrt{\frac{\log(4WN_t^2(a)/\alpha)}{2N_t(a)}}  \right\}. $$
The baseline confidence sequence used in \citet{jamieson} for BAI is given by:
$$C_t^{\text{base}}(a) =  \left\{ \hat\mu_t(a) \pm \sqrt{\frac{\log(\frac{405.5Wt^{1.1}}{\alpha}\log(\frac{405.5Wt^1.1}{\delta}))}{2N_t(a)}} \right\}. $$

For Hedged and Sub-B., we generalize them to the multi-arm case by dividing $\alpha = 0.05$ by $W$ to maintain our guarantees.

\subsection{Additional Details For Sampling Schemes in Multi-Arm Case.}
We provide the exact sampling schemes used for Section \ref{sec:experiments}'s experiments regarding threshold identification (THR) and best-arm identification (BAI). 

\textbf{Threshold Identification.} We set threshold $\xi = 0.5$. We use the sampling algorithm HDoC proposed by \citealp{kano2018good}: at time $t$, HDoC sample arm $a_t^\text{HDoC}$, defined as
$a_t^\text{HDoC} = \argmax_{a \in [K]} \hat\mu_t(a) + \sqrt{\frac{\log(t)}{2N_t(a)}}.$ We report the first time in which we have rejected $i$ regions as $\tau_i$, i.e. $\tau_i$ corresponds to the first time that we have labeled $i$ arms in Table \ref{table:THR}, and the final stopping time (i.e. $\tau_4$) in Table \ref{table:THR_shortened} in the main body of the text.

\textbf{Best Arm Identification.} Each arm $a$ in the BAI problem has a single corresponding hypothesis region given in Equation \eqref{eq:BAI}. We use the Lower-Upper Confidence Bound (LUCB) sampling scheme \cite{lucb}: at time $t$, LUCB samples two arms: arm $a_t^U = \max_{a \in [W]} \hat\mu_t$ and $a_t^L = \max_{a \in [W]\setminus a_t^U} \hat\mu_t  +  \sqrt{{\log\left( \frac{405.5At^{1.1}}{\alpha}\log(\frac{405.5At^{1.1}}{\alpha})\right)}/{2N_t(a)}}$. Our stopping time $\tau$ is defined as the first time in which we have eliminated all regions but one, where the last non-rejected region $R_{\text{BAI}}(a)$ corresponds to the best arm $a$.

\subsection{Runtime Testing.} While all tests were done in \texttt{R}, we test the runtime of each approach in Python. The grid size for Hedged is 100, 200, and 400 evenly spaced points on $[0,1]$ (i.e. fidelity $0.01, 0.005, 0.0025$ respectively), and tested in the manner of the description in Table 2 on page 50 of \citet{waudbysmith2022estimating}. All runtimes were done locally on an Apple M2 Pro Chip, 16 Gb of RAM, with no parallelization.

\subsection{Additional Simulations with Nonstandard Distributions}
To emphasize that PEAK performs well across all arbitrary data distributions over $[0,1]$, we provide additional results that test nonstandard distributions beyond exponential families. In particular, we test the following distributions, and provide additional results that demonstrate the robustness of PEAK's performance. 

\begin{itemize}
    \item \textbf{Mixture Distributions for a Single Stream.} In the single-stream case (Section 5.1), we test the mixture distribution, where observations are generated from $\text{Unif}(0,1)$, $\text{Beta}(1,1)$, and $\text{Bern}(0.5)$ with probability 1/3 each. Figure \ref{fig:mixture_conf_int} plots the width of the confidence sequence at time $t$ (i.e. the volume of the null hypothesis set we cannot reject at time $t$) formed by inverting PEAK, which empirically demonstrates that PEAK does not suffer (and is relatively better) than the other methods we test.

\begin{figure}[h]
    \centering
    \includegraphics[scale = 0.2]{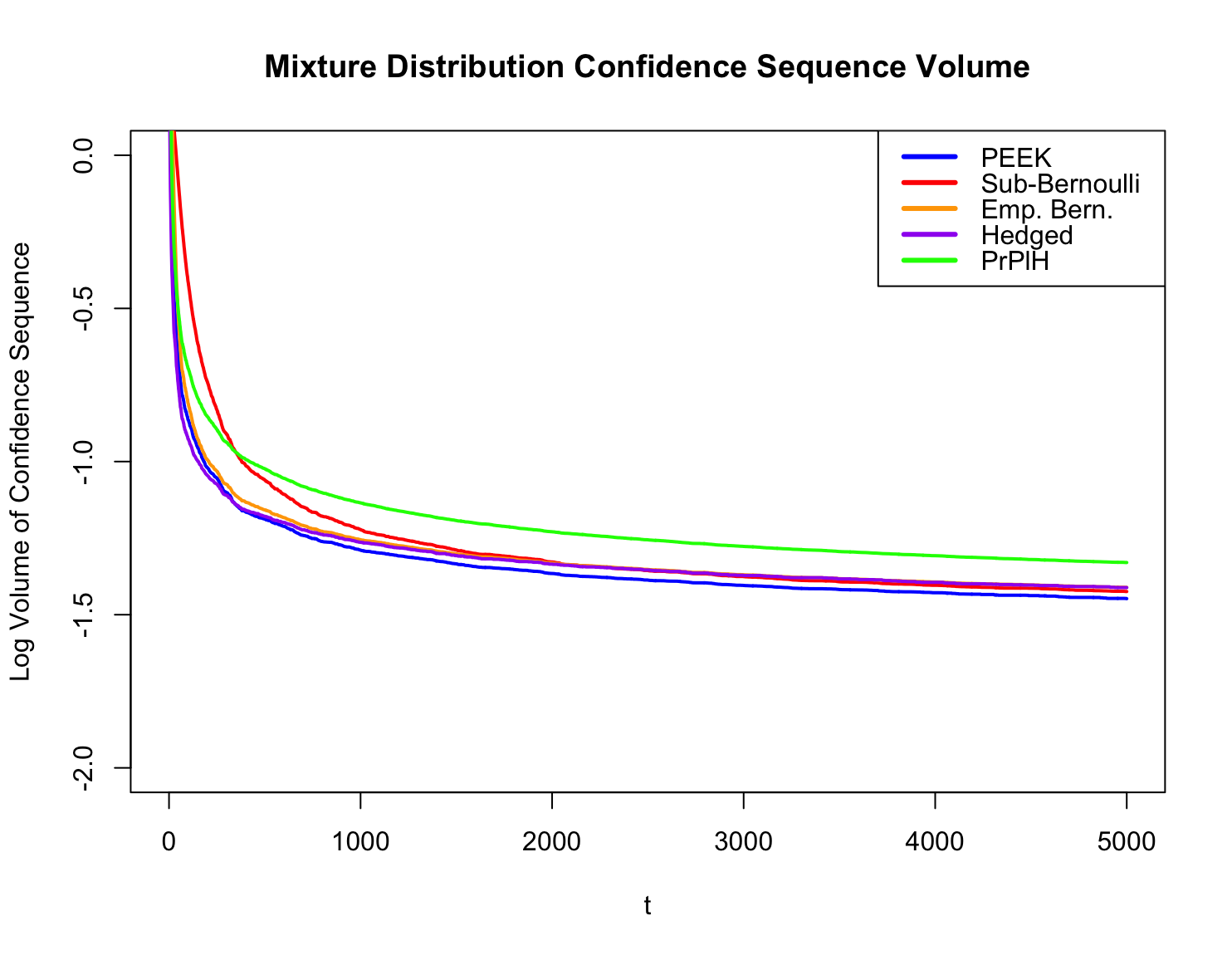}
    \caption{Confidence Sequence Widths in the Single Arm Setting for Mixture Distribution Arm}
    \label{fig:mixture_conf_int}
\end{figure}

    \item \textbf{Mixture Distribution for BAI.} For the multi-stream best-arm identification problem, we test distributions with adversarial heavy tails (outliers) in the best-arm identification setting by modifying our Beta distribution set-up. In our new simulations, we contaminate the distribution of arm 3 (second-best arm) with a point mass of $0.05$ at 1, and contaminate the distribution of arm 4 (best arm) with a point mass of $0.05$ at 0. To maintain the same $\mu$ across simulations, we choose Beta distributions with parameters $\alpha = 1$ for both arms 3 and 4, and $\beta = 0.43/0.52$, $\beta = 0.24/0.71$ for arms 3 and 4 respectively. The updated Table  provides empirical results regarding the stopping times, which align closely to the uncontaminated low-variance (uncontaminated Beta) setting.

\begin{table}[h]
\centering
\begin{tabular}{rrr}
  \hline
Type & Stopping Rule & Stopping Time \\ 
  \hline
  & PEAK & 708.52 $\pm$ 266.34 \\ 
Beta   & Base & 4686.1 $\pm$ 565.24 \\ 
    & Sub-B & 3175.3 $\pm$ 905.70\\ 
    & Hedged & \textbf{500.12 $\pm$ 214.08}\\ 
    \hline 
 & PEAK & \textbf{1318.14 $\pm$ 489.29} \\ 
Bern &   Base  &  4631.66 $\pm$ 896.51 \\
  &   Sub-B  & 4322.36 $\pm$ 3638.25\\
  & Hedged & 1734.64 $\pm$ 858.66 \\ 
\hline  
 & PEAK & 705.72 $\pm$ 284.76 \\ 
Beta-Contaminated &   Base  &  4680.92 $\pm$ 676.86 \\
  &   Sub-B  & 3376.22 $\pm$ 941.16\\
  & Hedged & \textbf{542.42 $\pm$ 246.24} \\ 
\hline  
\end{tabular}
\caption{Stopping Times for BAI, including contaminated distribution example.}
\end{table}
\end{itemize}

\newpage

\subsection{Additional Threshold Identification Results}

We now provide the average stopping time for identifying each arm sequentially, where $\tau_i$ is the first time in which $i$ arms are been labeled as above/below the threshold. 

\begin{table*}[h]\label{table:THR}
\centering
\begin{tabular}{rrrrrrr}
  \hline
Arm Type & Stopping Rule & $\tau_1$ & $\tau_2$ & $\tau_3$ & $\tau_4$  \\ 
  \hline
Beta & PEAK & \textbf{77.49 $\pm$ 34.62}  & \textbf{461.78 $\pm$ 185.64} & \textbf{554.52 $\pm$ 178.33} & \textbf{725.81 $\pm$ 218.49}  \\ 
    & Base & 253.71 $\pm$ 65.39 & 2391.64 $\pm$ 422.35 & 4476.97 $\pm$ 479.41& 4534.31 $\pm$ 482.94    \\ 
    & Sub-B & 199.98 $\pm$ 39.06 & 915.92 $\pm$ 219.18 & 1562.18 $\pm$ 230.75 & 1618.32 $\pm$ 229.90\\  
    & Hedged & \textbf{65.75  $\pm$ 28.73} & \textbf{316.74 $\pm$ 165.11} & \textbf{383.05 $\pm$ 161.48} & \textbf{479.16 $\pm$  205.97}
    \\ 
   \hline
Bernoulli & PEAK & \textbf{106.75 $\pm $ 70.29} & \textbf{905.64 $\pm$ 479.17} & \textbf{1329.32 $\pm$ 513.84} & \textbf{1678.08$\pm$ 666.39}  \\ 
    & Base & 253.83 $\pm$ 98.52 & 2479.43 $\pm$ 834.43 & 4731.87 $\pm$ 1107.71 & 4795.51 $\pm$ 1101.09  \\ 
    & Sub-B &  208.85 $\pm$ 64.17 & \textbf{944.26 $\pm$ 347.06} & {1624.80 $\pm$ 464.31} & \textbf{1686.81 $\pm$ 456.08} \\
    & Hedged & \textbf{104.77 $\pm$ 70.57} & 1212.75 $\pm$ 817.87 & \textbf{1582.60 $\pm$ 867.27} & 2241.18 $\pm$ 1092.49
    \\ 
    \hline 
\end{tabular}
\caption{Stopping Times $\tau$ for THR. Each $\tau_i$ represents the first time in which we have classified $i$ arms as above or below $\xi= 0.5$. We report averaged stopping times and one standard error over 100 simulated sample paths. Bolded values represent the best and second-best average run-time respectively. Across all simulations, each stopping time results in the correct conclusion (i.e. rejects the incorrect hypothesis).}
\end{table*}

\subsection{Visualization for Composite Hypothesis Testing for THR and BAI}

We visualize Examples 1 and 2 in Section \ref{sec:multi_stream}, providing simple intuition for the $W=2$ case. Below, we use $m_t^*$ to denote the global minimizer of $E_t(m)$, and $m_t^i$ to denote the minimizer of $E_t(m)$ for region $R_i$. These simple plots demonstrate the intuitive solutions implied by solving the KKT system of equations.

\begin{figure}[h]
    \centering
    \includegraphics[scale=0.5]{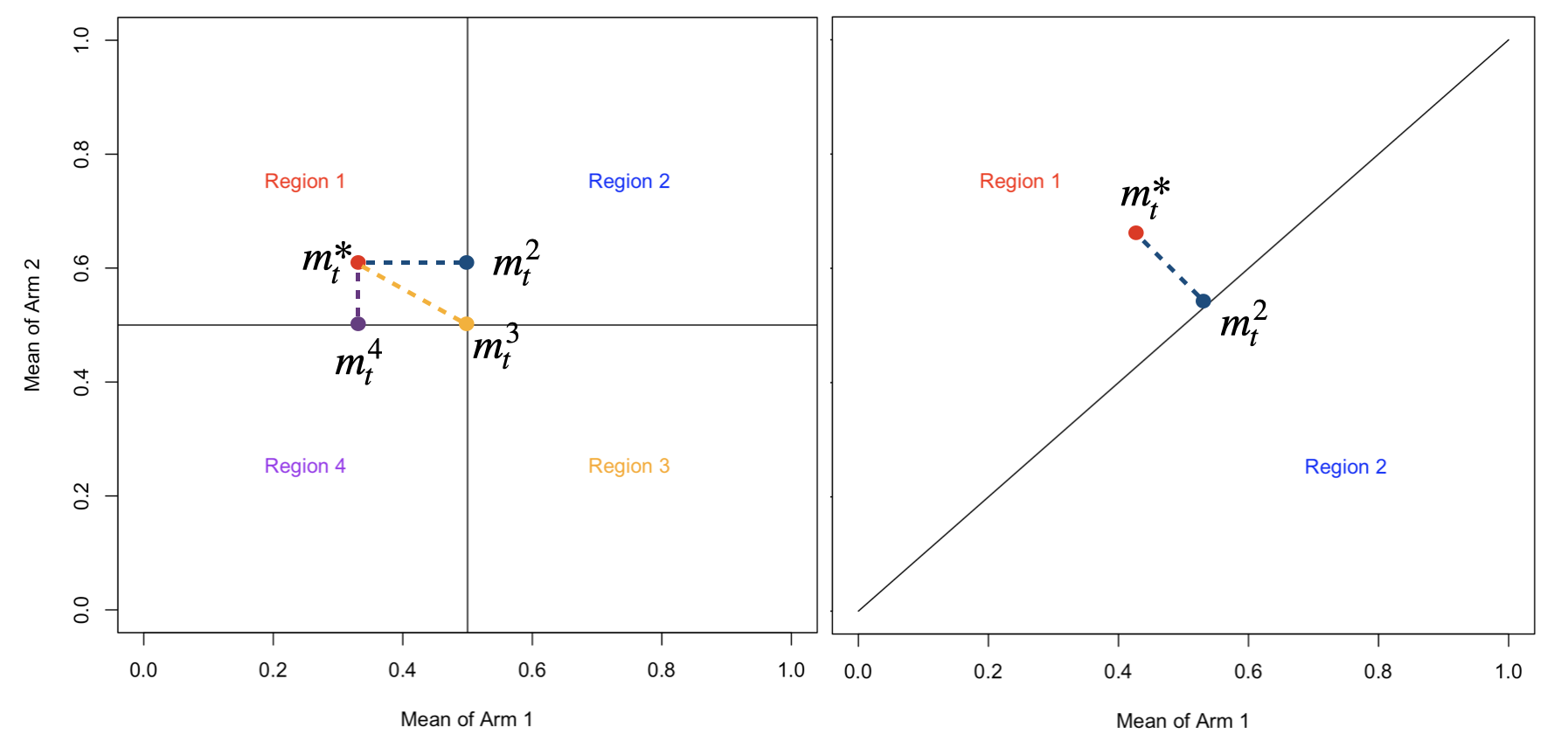}
    \caption{Visualization of Minima within each region at time $t$, obtained by projecting the global minimizer of $E_t(m)$ onto the regions implied by THR (left) and BAI (right) for $W=2$ case. In both plots, the current global minima at time $t$ is contained in Region 1, and projected to obtain the minima in all other regions.}
    \label{fig:thr_bai_viz}
\end{figure}

\end{document}